%% file: main.tex
\def\version{LNCS}
\newif \ifsubmission \submissionfalse 
\input{preamble}
\title{\sysname: A \sysnamefull \\ for Blockchain Wallets}
\anonfalse

\ifanon
  \author{}
  \authorrunning{Anonymous}
  \institute{}
\else
  \author{Jay Yu\inst{1,2} \and Shunfan Zhou\inst{3} \and Hang Yin\inst{3} \and Brian Seong\inst{4, 5}}
  \authorrunning{J. Yu et al.}
  \institute{
    Stanford University \email{jyu01@stanford.edu}
    \and
    Pantera Capital \email{jay@panteracapital.com}
    \and
    Phala Network \email{\{shelvenzhou,hangyin\}@phala.network}
    \and
    Polygon Labs \email{bseong@polygon.technology}
    \and
    Chainless Company \email{brian@chainless.company}
  }
\fi

\usepackage{algorithm}
\usepackage{algpseudocode}

\begin{document}
\maketitle

\input{abstract}

\input{Sections/all.tex}
\end{document}

%% file: preamble.tex
\input{Conferences/templates.tex}

\newif \ifcomments \commentsfalse
\newif \ifanon \anonfalse

\ifUSENIX \else \usepackage[table]{xcolor} \fi
\usepackage{xurl}
\usepackage{tabularx}
\usepackage{makecell}
\usepackage{booktabs}
\usepackage{xspace}
\usepackage{float}
\usepackage{hyperref}
\hypersetup{
    colorlinks=false,
    hidelinks
}
\usepackage[many]{tcolorbox}        
\usepackage{lipsum}
\usepackage{cleveref}
\usepackage{enumitem}
\setlist{nosep, topsep=1pt, leftmargin=*}
\usepackage{booktabs}
\usepackage{makecell}
\usepackage{multirow}
\usepackage[figuresright]{rotating}
\usepackage{array}
\usepackage{nicematrix}
\usepackage{tikz}
\usetikzlibrary{shapes, arrows.meta, positioning, calc, math, tikzmark, fit}

\usepackage{amsfonts,amsmath}
\ifLNCS\else\usepackage{amsthm}\fi
\ifACM\else\usepackage{amssymb}\fi

\definecolor{ForestGreen}{RGB}{34,139,34}

\ifcomments
    \newcommand{\jay}[1]{\textsf{\small{\color{red}{[Jay: {#1}]}}}}

\else

    \newcommand{\jay}[1]{}
\fi

\newcommand{\sysnamefull}{Provenanced Access Subaccount System\xspace}
\newcommand{\sysname}{PASS\xspace}
\newcommand{\sk}{{\sf sk}\xspace}
\newcommand{\pk}{{\sf pk}\xspace}

\newenvironment{icompact}{
  \begin{list}{$\bullet$}{
    \parsep 0pt
    \partopsep 0pt
    \topsep 1pt plus 0.5pt minus 0.5pt
    \itemsep 0pt
    \parskip 0pt plus 1pt
    \leftmargin 0.15in}}
  {\normalsize\end{list}}

\newenvironment{ncompact}{
  \begin{list}{\arabic{enumi}.}{
    \usecounter{enumi}
    \parsep 0pt
    \partopsep 0pt
    \topsep 1pt plus 0.5pt minus 0.5pt
    \itemsep 0pt
    \parskip 0pt plus 1pt
    \leftmargin 0.15in}}
  {\end{list}}

\tcbset{
    sharp corners,
    colback = white,
    before skip = 0.1cm,
    after skip = 0.2cm     
}                           

\newtcolorbox{boxA}{
    fontupper = \bf,
    boxrule = 1.5pt,
    colframe = black 
}

\iffull 
\else  \fi

\iffull 
\else  \fi

\ifLNCS\else

\iffull \newtheorem{theorem}{Theorem}[section]
\else \newtheorem{theorem}{Theorem} \fi
\newtheorem{corollary}{Corollary}[theorem]

\newtheorem{proposition}[theorem]{Proposition}

\theoremstyle{definition}

\theoremstyle{remark}

\fi

%% file: Conferences/templates.tex
\newif \iffull 
\newif \ifACM
\newif \ifUSENIX
\newif \ifIEEE
\newif \ifLNCS

\newif \ifCCS
\newif \ifSP
\newif \ifNDSS
\newif \ifCrypto
\newif \ifFC

\def\fullstring{full}
\def\ACMstring{ACM}
\def\USENIXstring{USENIX}
\def\IEEEstring{IEEE}
\def\LNCSstring{LNCS}

\def\CCSstring{CCS}
\def\SPstring{SP}
\def\NDSSstring{NDSS}
\def\Cryptostring{CRYPTO}
\def\FCstring{FC}

\ifx \version\fullstring \fulltrue \fi
\ifx \version\ACMstring \ACMtrue \fi
\ifx \version\USENIXstring \USENIXtrue \fi
\ifx \version\IEEEstring \IEEEtrue \fi
\ifx \version\LNCSstring \LNCStrue \fi

\ifx \version\CCSstring \ACMtrue \CCStrue \fi
\ifx \version\SPstring \IEEEtrue \SPtrue \fi
\ifx \version\NDSSstring \IEEEtrue \NDSStrue \fi
\ifx \version\Cryptostring \LNCStrue \Cryptotrue \fi
\ifx \version\FCstring \LNCStrue \FCtrue \fi

\iffull \input{Conferences/FullVersion/fullversion.tex} \fi
\ifACM \input{Conferences/ACM/acm.tex} \fi
\ifUSENIX \input{Conferences/USENIX/usenix.tex} \fi
\ifIEEE \input{Conferences/IEEE/IEEE.tex} \fi
\ifLNCS \input{Conferences/LNCS/lncs.tex} \fi

%% file: Conferences/LNCS/lncs.tex
\PassOptionsToPackage{dvipsnames}{xcolor}
\documentclass[runningheads]{Conferences/LNCS/llncs}

\ifsubmission \author{} \institute{} \fi

%% file: abstract.tex

\begin{abstract}
Blockchain wallets conventionally follow an ownership model where possession of a private key grants unilateral control. However, this assumption is brittle for emerging settings such as AI agent wallets, organizational custody, and enterprise payroll, where multiple actors must coordinate without exposing secrets or leaking internal activity.

\quad We present \sysname, a \textbf{Provenanced Access Subaccount System} that replaces role or identity-based control with provenance-based control: assets can only be used by subaccounts that can trace custody back to a valid deposit. A simple Inbox–Outbox mechanism ensures all external actions have verifiable lineage, while internal transfers remain private and indistinguishable from ordinary EOAs.  
  
\quad We formalize \sysname in Lean~4 and prove core invariants, including privacy of internal transfers, asset accessibility, and provenance integrity. We implement a prototype with enclave backends on AWS Nitro Enclaves and dstack Intel TDX, integrate with WalletConnect, and benchmark throughput across wallet operations. These results show that provenance-based wallets are both implementable and efficient. PASS bridges today’s gap between strict self-custody and flexible shared access, advancing the design space for practical, privacy-preserving custody.
\end{abstract}


%% file: Sections/all.tex

\input{Sections/introduction}
\input{Sections/background}

\input{Sections/pass-design}
\input{Sections/properties}
\input{Sections/practice}
\input{Sections/applications}

\input{Sections/future}
\input{Sections/conclusion}

\ifLNCS
\section*{Acknowledgements}
This work was supported in part by a grant from the Ethereum Foundation's 2025 Academic Grants Round. We thank Prof. Dan Boneh, Prof. Ari Juels, and the IC3 team for their guidance, feedback, and insightful discussions throughout this project. Their support significantly strengthened both the technical development and presentation of this work.
\fi


\iffull \bibliographystyle{plain} \fi
\ifACM \bibliographystyle{ACM-Reference-Format.bst} \fi
\ifUSENIX \bibliographystyle{plain} \fi
\ifIEEE \bibliographystyle{plain} \fi
\ifLNCS \bibliographystyle{Conferences/LNCS/splncs04.bst} \fi

\bibliography{references}

\ifSP \appendices
    \iffull
    \else
    \input{Sections/A-other-applications}
    \input{Sections/A-liveness-details}
    \input{Sections/A-complete-knowledge-details}

    \fi
\else
  \ifLNCS
    \appendix
    \input{Sections/A-formal-verification}
    \input{Sections/A-dstack-system}
    \input{Sections/A-benchmark-results}
    \input{Sections/A-design-extensions}
  \else
    \appendix
  \fi
\fi

%% file: Sections/introduction.tex
\section{Introduction}\label{sec:intro}

Cryptocurrency wallets conventionally assume that whoever holds the private key has permanent and total control of the address, and casts publicly readable, gas-consuming blockchain transactions. However, this model struggles to adapt to emerging blockchain scenarios:
\begin{ncompact}
  \item \textbf{AI and Smart Agent Transactions.} AI agents increasingly manage crypto assets \cite{patlan2025aiagentscryptoland,walters2025eliza} and require bounded authority and access guardrails for users' asset security.
  \item \textbf{Enterprise Payroll Privacy.} Organizations need to distribute salaries and departmental budgets to employee accounts without revealing sensitive details on-chain \cite{pertsev2019tornado,fincen2020ThresholdRule,auer2023bankingbis}.
  \item \textbf{Consumer-grade Central Limit Order Books.} On-chain order books face high gas costs and latency that prevent real-time trading at scale \cite{seong2025chainlessappsmodularframework,moosavi2021lissy}, and many consumers prefer not to manage their own keys \cite{brunner2021keystorage,yu2024walletchoices}.
\end{ncompact}

\noindent All three use cases demand support for \textbf{internal transfers} that are fast, low-cost (ideally free), and private, with most value movement occurring internally and on-chain transactions serving only for settlement.
Purely on-chain or custodial solutions cannot meet these requirements: on-chain transfers are slow, expensive, and public, while custodial models introduce trust assumptions and compromise privacy through full custodian visibility.
The core challenge involves \textbf{private shared state}—asset balances, transfer histories, and user subaccounts that must remain confidential while being collectively maintained. This requires TEEs for confidential computation with auditability, as well as designing appropriate access control models for this new computational paradigm.


In response to these needs, organizations and wallet providers have attempted to retrofit traditional access control paradigms—such as Role-Based Access Control (RBAC) and Attribute-Based Access Control (ABAC)—to govern internal transfers and shared wallet usage \cite{wallet_features_users2025,erinle2024shared_custodial}. This typically involves layering custom policy engines or custodial wrappers atop the standard single-owner wallet, specifying who may initiate or approve transfers under various conditions \cite{fireblocks_governance_policy_engine,turnkey_policy_quickstart}.
However, traditional access-control frameworks such as RBAC and ABAC assume an \emph{endogenous} setting, as users create and manage their own files, objects, or services, and policies are attached to these identities or roles. Blockchains are fundamentally \emph{exogenous}, where assets are created outside the wallet by smart contracts and transferred permissionlessly in. In this setting, identity or role-based rules are awkward to express and brittle to enforce. In a blockchain context, what matters instead is provenance, describing who deposited an asset, and how assets flow within a wallet.  

In this paper, we propose \sysname, a Provenanced Access Subaccount System that enforces control directly by provenance. The core mechanism is simple: an \textbf{Inbox} records external deposits, an \textbf{Outbox} signs and broadcasts withdrawals.
By leveraging a Trusted Execution Environment (TEE), all internal movements remain private and fee-less. Every unit of value inside the wallet has a verifiable lineage to an on-chain deposit, while externally a PASS wallet is indistinguishable from a standard EOA.
Our contributions include:
\begin{icompact}
    \item A provenance-based access model aligned with blockchain’s exogenous asset creation, enforced through an Inbox–Outbox design.
    \item A formal model in Lean~4 proving privacy, accessibility, and provenance integrity.
    \item A prototype with multi-vendor enclave backends (AWS Nitro, Intel TDX via dstack~\cite{zhou2025dstackzerotrustframework}), a web app frontend integrated with WalletConnect, together with throughput benchmarks.
    \item Application scenarios, including organizational custody, privacy-preserving transfers, and safe delegation to AI agents.
\end{icompact}

PASS thus bridges the gap between strict self-custody and flexible shared access, combining strong privacy guarantees with practical deployability.

%% file: Sections/background.tex
\section{Background and Related Work}\label{sec:background}
\input{Sections/Figures/sota-comparison}

Growing retail and institutional adoption of digital assets \cite{goldmansachs2021crypto} has driven diverse wallet designs \cite{mangipudi2023uncovering}. We summarize mainstream Ethereum Virtual Machine (EVM) wallet designs in Figure~\ref{tab:sota-wallet-comparison}.

\textbf{EOAs and Multisigs.} Externally Owned Accounts (EOAs) give unilateral, permanent control to whoever holds the private key \cite{buterin2014ethereum}. Multisignature wallets (e.g., Safe) distribute trust via $t$-of-$n$ approvals \cite{safe_docs_multisig}, common for treasuries and bridges \cite{tally_gnosis_safe,wormhole_guardians,forbes_bybit_hack}. Drawbacks include quorum overhead, symmetric authority, and public logic with attackable frontends \cite{forbes_bybit_hack}. In contrast, \sysname offers fine-grained, scoped permissioning without new contracts.

\textbf{Smart Contract Wallets and Account Abstraction.} Wallets like Argent and Loopring add programmability (social recovery, spend limits) via on-chain logic \cite{argent_docs,loopring_wallet_docs,cointelegraph2024smart}. ERC-4337 enables custom validation through \textsf{UserOperation} contracts \cite{erc4337_user_operation}, and EIP-7702 lets EOAs invoke this logic \cite{safe2024eip7702}. These increase flexibility but expose policies publicly and incur gas costs, limiting use cases like payroll. \sysname instead enforces policy privately in a TEE, minimizing on-chain footprint.

\textbf{Custodial Policies.} Custody providers (Fireblocks, Turnkey, Privy) use RBAC/ABAC engines with custom policy languages \cite{fireblocks_governance_policy_engine,turnkey_policy_quickstart}. \sysname replaces these with provenance: authority follows asset origin and history, akin to a UTXO model. This reduces policy-surface complexity while remaining composable with constraints like whitelists and blacklists. 

\textbf{TEEs and Liquefaction.} Trusted Execution Environments (TEEs) provide hardware-isolated contexts with remote attestation \cite{mckeen2013innovative,schneider2022sokhardwaresupportedtrustedexecution}. Liquefaction \cite{Austgen2024Liquefaction} keeps wallet keys in a TEE and enforces sub-policies via key encumbrance, ensuring privacy. \sysname advances this model by (i) proving invariants in Lean~4, (ii) simplifying enforcement to a provenance rule with Inbox–Outbox subaccounts, and (iii) working with a standard EOA without relying on Oasis smart contracts \cite{oasis_protocol}. We further validate by implementing PASS on AWS Nitro Enclaves and Intel TDX (via dstack), and showcase applications including AI wallets, payroll, supply chains, and DAO voting.

\textbf{Other Models.} Exchanges pool funds in omnibus addresses with off-chain ledgers \cite{venly_omnibus_vs_segregated_wallets}; mixers like Tornado Cash hide deposit–withdrawal links \cite{pertsev2019tornado}. \sysname similarly hides internal transfers while retaining a provenance log for auditability; outsiders see only deposits and withdrawals.

%% file: Sections/Figures/sota-comparison.tex
\begin{table*}[!ht]
  \centering
  \scriptsize
  \renewcommand{\arraystretch}{1.25}
  \begingroup
    \setlength{\tabcolsep}{5pt} 
    \begin{NiceTabular}{%
      >{\columncolor{white}\raggedright\arraybackslash}p{1.7cm}|
      >{\raggedright\arraybackslash}p{2.2cm}
      >{\raggedright\arraybackslash}p{2.2cm}
      >{\raggedright\arraybackslash}p{2cm}
      >{\raggedright\arraybackslash}p{2.0cm}
    }[color-inside]
      \CodeBefore
        \rowcolors{2}{gray!20}{}
      \Body
      \toprule
      \rowcolor{gray!50}
      \textbf{Model} & \textbf{Control basis} & \textbf{On-chain overhead} & \textbf{Internal Privacy} & \textbf{Compatibility} \\
      \midrule
      EOA & Single private key  & None & None & Native EOA \\
      \midrule
      Multisig & $t$-of-$n$ co-signers & High; quorum per action & None; logic public & Smart contract based \\
      \midrule
      ERC-4337 Smart Accounts
                     & On-chain policy logic                  & Medium to high; custom logic and gas   & None; logic public                     & Requires ERC-4337 chain support \\
      \midrule
      Custodians     & Custom RBAC or ABAC languages          & None; Off-chain enforcement                  & Custodian sees all activity            & API mediated \\
      \midrule
      Liquefaction   & TEE key encumbrance; policy trees      & Medium; Oasis contracts and gas & Strong; TEE hides subpolicies          & Requires Oasis deployment \\
      \midrule
      \textbf{PASS (ours)}
                     & \textbf{Provenance via Inbox–Outbox}  & \textbf{None; EOA compatible}          & \textbf{Strong; internal transfers invisible} & \textbf{Native EOA plus off-the-shelf TEEs} \\
      \bottomrule
    \end{NiceTabular}
  \endgroup
  \caption{Comparison of wallet architectures. PASS features provenance-based control, off-chain privacy, and native EOA compatibility.}
  \label{tab:sota-wallet-comparison}
\end{table*}

%% file: Sections/pass-design.tex
\section{System Model and Design}\label{sec:design}

\vspace{-0.8em}
\begin{figure}[!th]
\includegraphics[width=1\textwidth]{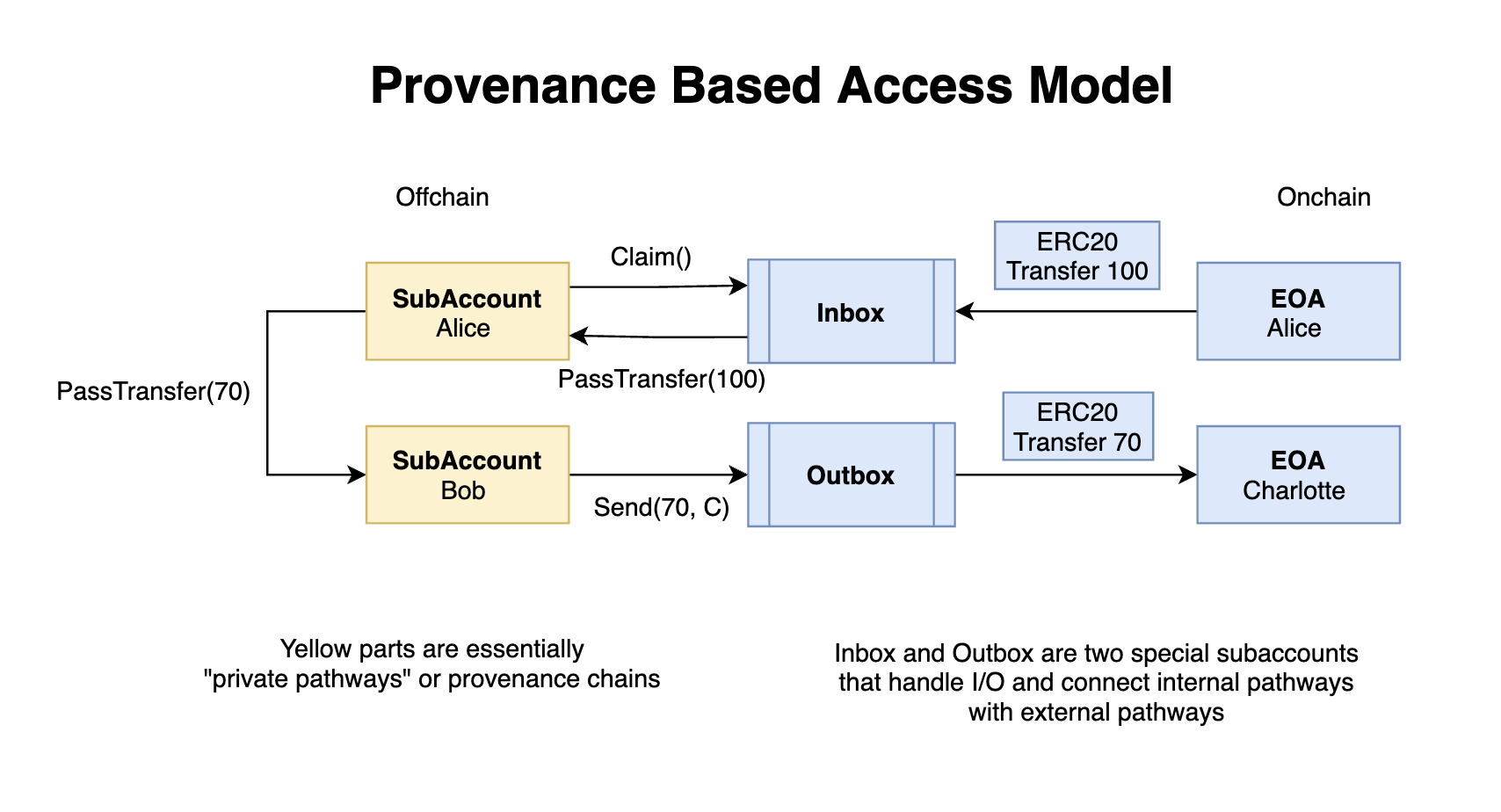}
\centering
\vspace{-0.8em}
\caption{Overview of \sysname. External deposits enter the \textbf{Inbox} and are claimed by subaccounts. Internal transfers are private and off-chain. Egress occurs via the \textbf{Outbox}, which signs and broadcasts on-chain transactions. All on-chain activity has clear provenance; internal movements remain private.}
\label{fig:pass-architecture}
\end{figure}
\vspace{-0.8em}

\subsection{High-level system model}\label{sec:design-highlevel}
\sysname is a multi-user wallet on a single EOA $(\sk,\pk)$. The wallet maintains a finite set of subaccounts $\mathcal{U}=\{u_1,\dots,u_n\}$ and a set of assets $\mathcal{A}$ (e.g., ERC-20). A subaccount $u\in\mathcal{U}$ has authority over an amount $x$ of asset $\alpha\in\mathcal{A}$ iff it can obtain a signature from the global $\sk$ via the enclave; $\sk$ resides only inside a TEE and is never directly accessible. We give the intuitive model here and formalize it in Section~\ref{sec:formal-system}.

Internally, \sysname tracks balances with a map $\mathcal{L}:\mathcal{U}\times\mathcal{A}\to\mathbb{Z}_{\ge 0}$ (accessed via $\mathcal{L}[u][\alpha]$), an \textbf{Inbox} multiset $\mathcal{I}$ of unclaimed deposits $(\alpha,x,m)$, an \textbf{Outbox} FIFO queue $\mathcal{O}$ of pending on-chain actions $(\alpha,x,\mathrm{extDst},\nu)$ together with a global nonce $\nu\in\mathbb{N}$, and an append-only provenance log $\mathcal{H}$. The core principles of PASS, as shown in Figure~\ref{fig:pass-architecture} are as follows:
\begin{ncompact}
  \item \emph{Access follows provenance.} A unit of value is usable by $u$ only if there exists a custody chain in $\mathcal{H}$ from an external deposit to $u$. Informally, $\textsf{Allow}(u,\alpha,x)\quad$ means balance $x$ is reachable for $u$ via provenance in $\mathcal{H}$. This implies \emph{consistency}: for every $\alpha$, the total internal balance equals the on-chain balance at $\pk$, i.e., $\sum_{u\in\mathcal{U}}\mathcal{L}[u][\alpha]=B_{\mathrm{on}}(\alpha)$. Incoming deposits $(\alpha,x,m)\in\mathcal{I}$ are ``wrapped'' and must be explicitly claimed before use by the sender address.

  \item \emph{Off-chain privacy.} Internal transfers $u\!\to\!v$ update only $\mathcal{L}$ and $\mathcal{H}$; they do not touch $\mathcal{O}$ or the on-chain balance $B_{\mathrm{on}}$, so internal flow is invisible and fee-less.

  \item \emph{First-come, first-served.} Outbox transactions are sequenced by a FIFO queue, which maintains and increments a global nonce $\nu$ and directly sends enclave-signed transactions to a RPC node to prevent nonce conflicts and replay attacks.\footnote{We assume that the TEE is able to run a Helios light client, as is the case in dstack Intel TDX, discussed and implemented in Section~\ref{sec:benchmarks}.}
\end{ncompact}

  \textbf{Organizational Ownership Model.} \sysname operates under an organizational ownership model, where an entity such as a company, individual, or DAO manages shared control over blockchain assets. We detail the threat model in Section~\ref{sec:properties}.

\textbf{Inbox–Outbox Gateway}
External senders transfer assets to $\pk$, creating Inbox entries $(\alpha,x,m)\in\mathcal{I}$. A user performs \emph{claim}$(u,\alpha,x,m)$ to move value into $\mathcal{L}[u][\alpha]$, conceptually “wrapping’’ the deposit. To withdraw or interact with a contract, a subaccount first performs an internal transfer to the \textbf{Outbox}, which enqueues $(\alpha,x,\mathrm{extDst},\nu)$; the enclave then signs the concrete transaction with $\sk$ and immediately broadcasts it, incrementing $\nu$. The FIFO discipline and global nonce serialize actions and prevent races and double-spends. Because subaccounts never receive pre-signed transactions, pre-signing misuse is mitigated.

\textbf{Composable policy scheme.}
Deposits are claims from $\mathcal{I}$; internal transfers require $\textsf{Allow}(u,\alpha,x)$ and update $(\mathcal{L},\mathcal{H})$; withdrawals are internal transfers to $\mathcal{O}$ followed by submission. Additional constraints (e.g., RBAC or ABAC checks, whitelists and blacklists) can be incorporated as conjuncts to \textsf{Allow} without changing the provenance-centric rule or its consistency and privacy effects.

\subsection{Formal Model}
\label{sec:formal-system}

We formalize \sysname as a transition system in Lean~4\footnote{Lean~4 model: \url{https://github.com/jayyu23/pass-lean4-proofs}}, an interactive theorem prover with a rich dependent type system \cite{lean_documentation}, suitable for policy verification (e.g., AWS Cedar \cite{cutler2024cedar}). In this section, we introduce the formal system we modeled in Lean, and in Section~\ref{sec:properties} we will discuss security invariants modeled. An extended discussion of our formal verification system can be found in Appendix~\ref{sec:formal-verification-extended}.

\textbf{State.} A wallet state is
$S=\bigl(\pk,\sk,\nu,\mathcal{I},\mathcal{O},\mathcal{L},\mathcal{H}\bigr)$
where $(\pk,\sk,\nu)$ are the EOA public key, private key (generated and held in the TEE), and global nonce; $\mathcal{I}$ and $\mathcal{O}$ are Inbox and Outbox; $\mathcal{L}$ is the internal balance map $\mathcal{L}[u][\alpha]\in\mathbb{Z}_{\ge 0}$; and $\mathcal{H}$ is the provenance history. For privacy arguments we also use the external view $W_{\mathrm{pub}}=(\pk,O,A)$, where $O$ contains the Outbox queue and nonce and $A$ records total on-chain balances by asset.

\textbf{Transitions.}
Operations $\textsf{op} \in T$ update $S$ with type-checked rules:
\emph{InboxDeposit} adds an unclaimed $(\alpha,x,m)$ to $\mathcal{I}$; \emph{ClaimInbox} moves $(\alpha,x,m)$ from $\mathcal{I}$ to $\mathcal{L}[u][\alpha]$ and appends to $\mathcal{H}$; \emph{InternalTransfer} moves value between subaccounts; \emph{Withdraw} debits $\mathcal{L}[u][\alpha]$, enqueues $(\alpha,x,\mathrm{extDst},\nu)$ in $\mathcal{O}$, and logs; \emph{SignGSM} authorizes domain-scoped signable messages. A representative rule is:
\[
\bigl(\pk,\sk,\nu,\mathcal{I},\mathcal{O},\mathcal{L},\mathcal{H}\bigr)\xrightarrow{\textsf{transfer}(\alpha,x,u_s,u_r)}\bigl(\pk,\sk,\nu,\mathcal{I},\mathcal{O},\mathcal{L}',\mathcal{H}'\bigr)
\]
with
\[
\begin{aligned}
&\text{if } \mathcal{L}[u_s][\alpha]<x \text{ or } \neg\textsf{checkAllow}(u_s,\alpha,x,\mathcal{H}) \text{ then } \bot;\\
&\mathcal{L}'[u_s][\alpha]=\mathcal{L}[u_s][\alpha]-x,\quad \mathcal{L}'[u_r][\alpha]=\mathcal{L}[u_r][\alpha]+x;\\
&\mathcal{H}'=\mathcal{H}\ \big|\ \bigl(\textit{op}=\texttt{transfer},\alpha,x,u_s,u_r\bigr).
\end{aligned}
\]

Here \textsf{checkAllow} queries provenance in $\mathcal{H}$ and can be extended with auxiliary constraints such as user whitelists or blacklists, without changing the core rule. Figure~\ref{fig:pass-symbolic-create} provides a formal specification of the methods in a PASS wallet.

\input{Sections/Figures/PASS-functions}

\textbf{Provenance Attestation.} The TEE can generate cryptographic attestations over the integrity of $\mathcal{H}$ by signing its digest: $\sigma = \textsf{Sign}_{\sk}(\textsf{Hash}(\mathcal{H}))$. This enables external verification of custody chains for compliance auditing and regulatory requirements without revealing internal transfer details. Attestations can also cover other system components like the global nonce $\nu$ for freshness proofs or balance states for solvency verification, following EIP-712 structured data standards for interoperability.

\subsection{Asset Model and General Signable Messages}\label{sec:asset-model}
\sysname provides an authorization scheme over a heterogeneous set of assets $\mathcal{A}$, including fungible tokens, such as ether and ERC-20 tokens, ERC-721 non-fungible tokens, and ``general signable messages" across a variety of DApps for sign-in and identity verification. Each asset $\alpha \in \mathcal{A}$ has a well-defined custody path in the provenance log $\mathcal{H}$, and a subaccount $u\in\mathcal{U}$ has usable balance $x$ of $\alpha$ iff $\textsf{Allow}(u,\alpha,x)$ holds.

\textbf{Token assets.} 
For token assets such as ether, ERC-20, and ERC-721, deposits are captured as Inbox entries $(\alpha,x,m)\in\mathcal{I}$. Claims allocate balance into $\mathcal{L}[u][\alpha]$, internal transfers update $\mathcal{L}$ and $\mathcal{H}$ privately, and withdrawals enqueue $(\alpha,x,\mathrm{extDst},\nu)$ into $\mathcal{O}$ for enclave signing and on-chain broadcast, ensuring consistency where:
\[
\forall \alpha \in \mathcal{A}_{\mathrm{tokens}}:\quad \sum_{u\in\mathcal{U}}\mathcal{L}[u][\alpha] \;=\; B_{\mathrm{on}}(\alpha,\pk).
\]

In particular, ERC-721 tokens are modeled with unique identifiers, where each token ID is treated as a distinct asset $\alpha_{tokenId}$ with unitary balance ($x\in\{0,1\}$), ensuring uniqueness and preventing double-ownership under provenance.

\textbf{General Signable Messages (GSMs).}
In addition to on-chain assets, wallets often sign arbitrary messages, such as for EIP-4361 \cite{siwe_eip4361} logins and off-chain DAO votes \cite{snapshot_voting}. \sysname models these as a special class $\mathcal{A}_{\mathrm{GSM}}$ of virtual assets. Each GSM asset is scoped to a domain $\mathrm{dom}$, so holding $\alpha=(\mathrm{dom})$ represents authority to sign messages from that domain. A signature operation:
$
\textsf{SignGSM}(u,\mathrm{dom},m) \mapsto \sigma$ is permitted iff $\textsf{Allow}(u,\mathrm{dom},1)$ holds in $\mathcal{H}$. At key generation, the wallet creator is endowed with \emph{default signing authority} over all domains, represented as an initial allocation of every $\alpha=(\mathrm{dom})\in\mathcal{A}_{\mathrm{GSM}}$ to the root subaccount. These domain assets can be internally transferred to delegate signing authority. Thus provenance of GSMs begins at the owner and flows through transfers, providing a well-defined provenance path in $\mathcal{H}$ as with token assets.

By treating both tokens and GSMs as types of elements of $\mathcal{A}$, \sysname provides a unified provenance rule for subaccount authorization. A subaccount cannot spend an ERC-20 token, transfer an NFT, or sign a domain message unless it possesses the corresponding asset under $\mathcal{L}$ and $\mathcal{H}$. This unified authorization scheme allows for fine-grained delegation, (e.g., giving $u$ rights over \texttt{example.com} while withholding ERC--20 funds) and consistent enforcement across on-chain and off-chain actions.

%% file: Sections/Figures/PASS-functions.tex
\begin{figure}[!t]
\centering
\fbox{%
\begin{minipage}{0.94\linewidth}
\scriptsize

\noindent
\textbf{PASS Wallet Management Functionality} 
$\;\mathcal{F}_{\text{PASS}}^{\Sigma}$ for wallet $\mathcal{W}$

\noindent
\textbf{Global State:}
\begin{itemize}[label={},noitemsep,nolistsep,leftmargin=1.5em]
  \item $\mathit{pk}$: public address (EOA)
  \item $\mathit{sk}$: private key (in TEE)
  \item $\nu$: global nonce for on-chain transactions
  \item $\mathcal{I}$ (\textsf{Inbox}): set of unclaimed deposits $(\alpha: \text{asset}, x: \text{amount}, m: \text{depositId})$
  \item $\mathcal{O}$ (\textsf{Outbox}): FIFO queue $(\alpha , x, \textit{extDst}, \nu)$
  \item $\mathcal{L}$ (\textsf{Asset Ledger}): internal map $\mathcal{L}[u][\alpha]\in\mathbb{Z}_{\ge 0}$, where $u$ is each user's ID in the system
  \item $\mathcal{H}$ (\textsf{Provenance History}): a list of transaction records appended whenever \textbf{ClaimInbox}, \textbf{InternalTransfer}, or \textbf{Withdraw} is invoked
\end{itemize}

\medskip

\noindent
\textbf{CreatePassWallet} $(\mathit{creator})$:
\begin{itemize}[label={},noitemsep,nolistsep,leftmargin=3em]
  \item $\mathcal{I} \!\gets\! \varnothing,\ \mathcal{O} \!\gets\! \varnothing,\ 
        \mathcal{L} \!\gets\! \{\};\ \nu\!\gets 0$
  \item ($\mathit{pk} ,\ \mathit{sk}) \!\gets\! \text{TEE\_KeyGen}()$
  \item Return new $\mathcal{W}$ with 
        $S = \bigl(\mathit{pk}, \mathit{sk}, \nu, \mathcal{I}, \mathcal{O}, \mathcal{L}, \mathcal{H}\bigr)$
        and reference to $\mathit{creator}$
\end{itemize}

\noindent
\textbf{InboxDeposit} $(\alpha, x, m)$ from external $P$:
\begin{itemize}[label={},noitemsep,nolistsep,leftmargin=3em]
  \item $\mathcal{I} \;\gets\; \mathcal{I}\ \cup\ \{(\alpha, x, m)\}$
  \item Return $\texttt{success}$
\end{itemize}

\noindent
\textbf{ClaimInbox} $(\alpha, x, \textit{depositId})$ from $u$:
\begin{itemize}[label={},noitemsep,nolistsep,leftmargin=3em]
  \item If no $(\alpha, x, m)$ with $\textit{depositId}$ in $\mathcal{I}$, return $\bot$
  \item Remove $(\alpha, x, m)$ from $\mathcal{I}$
  \item $\mathcal{L}[u][\alpha] \gets \mathcal{L}[u][\alpha] + x$
  \item $\mathcal{H} \gets \mathcal{H} \,|\, (\textit{op}=\texttt{claim}, \alpha, x, \textit{depositId}, u)$
  \item Return $\texttt{success}$
\end{itemize}

\noindent
\textbf{CheckAllow} $(u_{\mathrm{send}}, \alpha, x, \mathcal{H})$:

\begin{itemize}[label={},noitemsep,nolistsep,leftmargin=3em]
\item If $\mathcal{H}$.GetProvenance($\alpha$, $u_{\mathrm{send}}) \ge x$,
\item return $\texttt{true}$; else return $\texttt{false}$
\end{itemize}

\noindent
\textbf{InternalTransfer} $(\alpha, x, u_{\mathrm{send}}, u_{\mathrm{recv}})$:

\begin{itemize}[label={},noitemsep,nolistsep,leftmargin=3em]
\item If $\mathcal{L}[u_{\mathrm{send}}][\alpha]<x \; \text{or } \neg \textsf{checkAllow$(u_{\mathrm{send}}, \alpha, x, \mathcal{H}$})$, return $\bot$
\item $\mathcal{L}[u_{\mathrm{send}}][\alpha] \gets \mathcal{L}[u_{\mathrm{send}}][\alpha] - x$
\item $\mathcal{L}[u_{\mathrm{recv}}][\alpha] \gets \mathcal{L}[u_{\mathrm{recv}}][\alpha] + x$
\item $\mathcal{H} \gets \mathcal{H} \,\bigl|\bigr.\, \bigl(\textit{op}=\texttt{transfer}, \alpha, x, u_{\mathrm{send}}, u_{\mathrm{recv}}\bigr)$
\item Return $\texttt{success}$
\end{itemize}

\noindent
\textbf{Withdraw} $(\alpha, x, u_{\mathrm{send}}, \textit{extDst})$:
\begin{itemize}[label={},noitemsep,nolistsep,leftmargin=3em]
    \item If $\mathcal{L}[u_{\mathrm{send}}][\alpha]<x \; \text{or } \neg \textsf{checkAllow$(u_{\mathrm{send}}, \alpha, x, \mathcal{H})$}$, return $\bot$
    \item $\mathcal{L}[u_{\mathrm{send}}][\alpha] \gets \mathcal{L}[u_{\mathrm{send}}][\alpha] - x$
    \item $\mathcal{O} \gets \mathcal{O} \cup \{(\alpha, x, \textit{extDst}, \nu)\}$
    \item $\mathcal{H} \gets \mathcal{H} \ | \ (\textit{op}=\texttt{withdraw}, \alpha, x, u_{\mathrm{send}}, \textit{extDst})$
    \item Return $\texttt{success}$
\end{itemize}

\noindent
\textbf{SignGSM} $(\mathit{dom}, \mathit{msg}, u)$:
\begin{itemize}[label={},noitemsep,nolistsep,leftmargin=3em]
  \item If $u$ lacks sign-right for $(\mathit{dom}, \mathit{msg})$, return $\bot$
  \item $\sigma \gets \textsf{Sign}(\mathit{sk}, \mathit{msg})$
  \item Return $\sigma$
\end{itemize}

\noindent
\textbf{(Outbox Processing)} (periodic or on demand):
\begin{itemize}[label={},noitemsep,nolistsep,leftmargin=3em]
  \item If $\mathcal{O}\neq\varnothing$, dequeue $(\alpha, x, \textit{extDst}, \nu)$
  \item Build tx $\tau$: transfer $x$ of $\alpha$ to $\textit{extDst}$, nonce $=\nu$
  \item $\sigma \gets \textsf{Sign}(\mathit{sk},\, \tau)$; broadcast $(\tau,\sigma)$
  \item $\nu \!\gets\!\nu + 1$
  \item Return $\texttt{broadcasted}$ + txHash
\end{itemize}

\end{minipage}}

\caption{\small Formal Model of PASS wallet functionality, including \textbf{CreatePassWallet}
for initializing a fresh instance.
On creation, $\mathit{pk}$ is set to a chosen EOA address $\mathit{eoa}$, the TEE
generates a hidden private key $\mathit{sk}$, and global data
structures are initialized. 
External deposits enter $\mathcal{I}$ (\textsf{Inbox}), 
while subaccount transfers update $\mathcal{L}$ (\textsf{Ledger}) off-chain. 
Withdrawals queue items in $\mathcal{O}$ (\textsf{Outbox}), 
and GSM operations are signed using $\mathit{sk}$.}
\label{fig:pass-symbolic-create}
\end{figure}

%% file: Sections/properties.tex
\section{Properties of \sysname}\label{sec:properties}

\subsection{Threat Model}

\sysname operates under an organizational ownership model where a single entity (company, DAO, or individual) manages the EOA $(pk, sk)$ and TEE infrastructure. This model is typical for use cases described in Section~\ref{sec:intro}. Protected assets include the private key $sk$, balance ledger $\mathcal{L}$, provenance log $\mathcal{H}$, and access control policy \textsf{Allow}($u, \alpha, x$).
In our model, we make the following trust assumptions:
\begin{icompact}
    \item \textbf{TEE integrity.} Hardware isolation is sound; $sk$ never leaves the enclave. Remote attestation guarantees that the
    enclave runs the genuine \sysname binary.
    \item \textbf{Cryptography.} Standard primitives (signatures, hash functions, PRFs) are secure against polynomial-time
    adversaries.
    \item \textbf{Governance majority.} The deploying entity holds majority control in the governance mechanism (via DAO vote or admin privileges), ensuring it can invoke the blockchain-governed KMS for MPC-based key recovery and trigger fallback contracts when needed.  
\end{icompact}

\textbf{Adversary Model.}
An adversary $\mathcal{A}$ may control the host OS and network,
observe or modify all communication between enclaves, and submit arbitrary transaction requests. In particular, $\mathcal{A}$ may attempt to:
\begin{icompact}
    \item Exploit software bugs in the \sysname implementation or TEE
    runtime,
    \item Forge or replay transactions to break consistency of $\mathcal{H}$,
    \item Subvert provenance by violating
    \textsf{Allow}$(u,\alpha,x)$
\end{icompact}

\textbf{Threats Out of Scope.}
Physical TEE compromise, supply-chain attacks, and malicious operators with unrestricted access (excluded due to organizational security controls) are regarded out of scope because the organizational model assumes controlled facilities and vetted hardware.

\textbf{Security Goals.} Given the assumptions above, \sysname guarantees:
\begin{icompact}
    \item \textbf{Integrity:} $\sum_{u \in \mathcal{U}} \mathcal{L}[u][\alpha] \leq
    \textsf{extDep}(\alpha)$ for all $\alpha$, i.e., no assets can be
    created or double-spent internally.
    \item \textbf{Correctness:} If \textsf{Allow}$(u,\alpha,x)$ holds,
    then $u$ can eventually realize $x$ via \textsf{Withdraw}.
    \item \textbf{Privacy on-chain:}
    internal transfer operations do not affect
    $\mathsf{externalTrace}(W)$, and to external observers, a \sysname
    wallet is indistinguishable from an EOA.
    \item \textbf{Liveness:} For any authorized $\mathsf{op}\in\mathsf{Ops}$, if one enclave and quorum remain honest, then $\exists\tau$ s.t. $\mathsf{Exec}(c,\mathsf{op})=\top$ within $\tau$.
\end{icompact}

\subsection{Privacy and Safety}
After establishing the threat model of \sysname, we will address the central privacy and safety guarantees in our formal verification model. We sketch three central lemmas here, and create mechanized proofs in Lean, discussed in Appendix~\ref{sec:formal-verification-extended}.

\begin{proposition}[Internal Transfer Privacy]
    If two public states $W_1,W_2$ share the same $\pk$, identical Outbox $(T,\nu)$, and equal per-asset totals in $A$, then
\[
\mathrm{externalTrace}(W_1)=\mathrm{externalTrace}(W_2).
\]
\end{proposition}

\emph{Proof Sketch.}
\textsf{InternalTransfer} updates only $(\mathcal{L},\mathcal{H})$, never $\mathcal{O}$ or total on-chain balances. \textsf{OutboxProcessing} in Figure~\ref{fig:pass-symbolic-create} emits on-chain actions solely from $O$ (lines 2–4), hence the trace depends only on $\mathcal{O}$.

\emph{Corollary: EOA indistinguishability.}  
If $W_{\mathrm{pass}}$ and $W_{\mathrm{eoa}}$ share the same $\pk$, Outbox, and per-asset totals, then
\[
\mathrm{externalTrace}(W_{\mathrm{pass}})
=\mathrm{externalTrace}(W_{\mathrm{eoa}}).
\]
Thus, from an external observer’s perspective, a \sysname account is indistinguishable from a standard EOA: internal transfers leave no observable footprint.

\begin{proposition}[Asset Accessibility]
    For any asset $\alpha$, letting $L=\{u\mid \mathcal{L}[u][\alpha]>0\}$,
$$\sum_{u\in L}\mathcal{L}[u][\alpha]=\mathrm{totalBalance}(\alpha),
$$
and each $u\in L$ can spend its full balance $\mathcal{L}[u][\alpha]$
if and only if the following predicate holds:
\[
\textsf{checkAllow}(u,\alpha,\mathcal{L}[u][\alpha],\mathcal{H}).
\]
\end{proposition}

\emph{Proof Sketch.} \textsf{ClaimInbox} credits $\mathcal{L}$; \textsf{InternalTransfer} preserves the sum over $u$; \textsf{Withdraw} debits $\mathcal{L}$ and enqueues to $\mathcal{O}$, which \textsf{OutboxProcessing} realizes on-chain. \textsf{checkAllow} enforces reachability from provenance in $\mathcal{H}$.

\begin{proposition}[Provenance Integrity]
   Let $\mathrm{extDep}(\alpha)$ be net external deposits (deposits minus realized withdrawals). For all reachable states,
\[
\sum_{u \in \mathcal{U}}\mathcal{L}[u][\alpha]\ =\ \mathrm{extDep}(\alpha).
\]
\end{proposition}

\emph{Proof Sketch.} Induct on transitions.
\textsf{InboxDeposit} increases $\mathrm{extDep}$ only, and \textsf{ClaimInbox}
moves already-deposited value into $\mathcal{L}$.
\textsf{InternalTransfer} conserves the sum over $u$.
\textsf{Withdraw} decreases the sum, and after \textsf{OutboxProcessing}
broadcasts, $\mathrm{extDep}$ is reduced accordingly.

\textbf{General signable messages (GSM).}
Model GSM rights as a domain-scoped virtual asset class $\mathcal{A}_{\mathrm{GSM}}$. Let $\alpha=(\mathrm{dom})\in\mathcal{A}_{\mathrm{GSM}}$ and permit
\[
\textsf{SignGSM}(u,\mathrm{dom},m)\ \ \text{iff}\ \ \textsf{Allow}(u,\alpha,1).
\]
Then (P2)–(P3) lift to GSM: there is always at least one domain signer, and GSM “balances” obey provenance conservation under transfers in $\mathcal{H}$.

\subsection{Liveness}

We define liveness as the property that any authorized operation remains eventually executable, despite enclave crashes, vendor compromise, or partial governance failure. Our model assumes deployment of \sysname access logic on dstack, implemented in Section~\ref{sec:practice} and detailed in Appendix~\ref{sec:dstack}.

Let $\mathcal{E}$ denote active enclaves, $\mathcal{C}$ containerized wallet backends, and $\mathcal{P}$ governance policies.  
Each container $c\in\mathcal{C}$ stores sealed state
\[
\sigma_c = \langle S,k_c,\mathsf{meta}\rangle
\quad\text{with}\quad
S=(\pk,\sk,\nu,\mathcal{I},\mathcal{O},\mathcal{L},\mathcal{H}),
\]
where $k_c=\mathsf{Derive}(\mathcal{K},c)$ is derived from an MPC-managed Key Management Service (KMS) and $\mathsf{meta}$ records attestation.  
An operation $\mathsf{op}\in\mathsf{Ops}$ (deposit claim, transfer, withdrawal) executes as:
$\mathsf{Exec}(c,\mathsf{op}) = 
\top \ \text{iff}\ \mathsf{Attested}(c)\wedge \Gamma(S)$,
where $\mathsf{Attested}(c)$ checks $\sigma_c$ against registry $\mathcal{R}$ and $\Gamma$ is the quorum predicate over $\mathcal{P}$. To ensure vendor-agnostic portability, state migration $\mu:\mathcal{E}\to\mathcal{E}$ satisfies
\[
\forall e_i,e_j\in\mathcal{E},\quad
e_j.\mathsf{decrypt}(\sigma_{c,e_i},k_c)=\sigma_{c,e_i},
\]

\begin{theorem}[Liveness Theorem]
If (i) the MPC KMS tolerates $f_{\mathrm{KMS}}$ faults in a $t$-of-$n$ threshold,  
(ii) governance quorums are $f_{\mathrm{gov}}$-resilient, and  
(iii) at least one enclave is uncompromised,  
then every $\mathsf{op}\in\mathsf{Ops}$ executes within bounded time $\tau$:
\[
\forall \mathsf{op}\in\mathsf{Ops},\ \exists \tau\ge0:\ \mathsf{Exec}(c,\mathsf{op})=\top.
\]
\end{theorem}

\textbf{Proof Sketch.}
Progress follows from three guarantees:  
(1) the MPC KMS ensures $k_c$ is always recoverable under threshold availability;  
(2) $\Gamma$ encodes quorum rules that eventually relax after timeout $\Delta t$, ensuring some subset $S$ satisfies $\Gamma(S)=\top$;  
(3) migration $\mu$ allows any enclave to re-host $\sigma_c$, preventing vendor lock-in or enclave failure.  
Algorithm~\ref{alg:recover-and-execute} formalizes recovery and execution under these assumptions.

\begin{algorithm}[H]
\caption{\textsf{RecoverAndExecute}$(c,\mathsf{op})$}
\label{alg:recover-and-execute}
\footnotesize
\begin{algorithmic}[1]
\State Verify enclave $e$ attestation $Q_e$ against $\mathcal{R}$.
\State Derive $k_c \gets \mathsf{Derive}(\mathcal{K},c)$; decrypt $\sigma_c$.
\If{$\mathsf{CompromiseDetected}(e')$} \State Rotate $k_c$ and re-encrypt $\sigma_c$. \EndIf
\State Ensure $\Gamma(S)=\top$ for $S\subseteq\mathcal{P}(c)$.
\State Execute $\mathsf{op}$; append to $\mathcal{H}$.
\end{algorithmic}
\end{algorithm}

\textbf{Key Properties.}
(1) \emph{Portability:} any attested enclave can recover $\sigma_c$, independent of vendor, given that $k_c$ is managed by KMS; 
(2) \emph{Monotonic Recovery:} recovery steps form $\rho:\mathcal{S}\to\mathcal{S}$ with $\rho(s)\succeq s$; 
(3) \emph{Domain Continuity:} certificates remain bound to $\mathcal{R}(c)$ under migration $\mu(c)$, preventing hijack or replay.

Implementation details of container orchestration, attestation, and KMS protocols are given in Appendix~\ref{sec:dstack}.

%% file: Sections/practice.tex
\section{Practical Implementations}\label{sec:practice}

\subsection{Prototype Implementations}

To demonstrate feasibility, we have implemented a prototype of \sysname that captures the core architectural components.\footnote{Prototype and deployment instructions can be found here: \url{https://github.com/jayyu23/pass-wallet-app}} The prototype consists of:

\begin{icompact}
    \item \textit{Provenanced Access Model}, the core \sysname architecture outlined in Section~\ref{sec:design} with the Inbox–Outbox structure, asset ledgers, and provenanced rule as a Rust executable.
    \item \textit{Web Application Frontend}, a Typescript interface that allows users to deposit assets, initiate internal transfers, generate on-chain transactions, and sign arbitrary messages on Ethereum Sepolia testnet via WalletConnect v2.
    \item \textit{Multi-vendor TEE deployment}, we deploy the \sysname enclave executable in both AWS Nitro Enclaves and Intel TDX with the dstack framework.
\end{icompact}



We evaluate these two enclave deployments as complementary paths with distinct tradeoffs. \textbf{AWS Nitro Enclaves} simplify deployment through AWS integration and attestation, but concentrate trust in a single vendor and store the provenance ledger externally, potentially weakening liveness and asset accessibility guarantees. \textbf{Intel TDX with dstack} provides a portable, open-source runtime across diverse hardware with decentralized operators and MPC-based key management, preserving full security properties by executing both Inbox–Outbox logic and provenance ledger inside the enclave, though at the cost of added governance complexity for operator authorization and policy enforcement.

\textbf{Blockchain Connectivity.}
To ensure reliable interaction, \sysname embeds Helios~\cite{helios_lightclient}, a consensus-verified light client within the TEE, enabled via dstack deployment. 
Helios verifies blockchain data against consensus proofs rather than trusting external RPC endpoints, removing reliance on centralized providers while preserving real-time performance. 
This guarantees that \sysname's Outbox operates only on verified state, even if external network providers are compromised.

\subsection{Performance Benchmarks}\label{sec:benchmarks}

\textbf{Methodology.} To evaluate \sysname, we benchmarked its performance using both AWS Nitro Enclaves and Intel TDX as secure backend environments. We selected a suite of representative wallet operations—wallet creation, balance queries, claims, transfers, withdrawals, provenance queries, and a comprehensive end-to-end workflow\footnote{Complete sequence: deposit $\rightarrow$ claim $\rightarrow$ transfer $\rightarrow$ withdraw $\rightarrow$ process outbox}. To further stress-test \sysname, we conducted batch experiments involving $10,000$ consecutive deposit and claim operations, as well as transfer-then-withdraw scenarios utilizing multi-threaded processing.
Each operation was run $100$ times over $10$ trials per environment, measuring throughput in operations per second (ops/sec).
Both environments ran on identical virtual machines (4 cores, 2.5 GHz, 16 GB RAM) to ensure a fair comparison. Table~\ref{tab:benchmarks} summarizes our results, with full details available in Appendix~\ref{appendix:benchmark-results}.

\begin{table*}[h]
\centering
\caption{Internal \sysname operation throughput (ops/sec) across AWS Nitro Enclaves and Intel TDX. All operations are in-memory state changes within the TEE with no blockchain interaction. Each value represents the average of 100 operations repeated 10 times per environment.}
\label{tab:benchmarks}
\footnotesize
\begin{tabular}{lrr}
\toprule
\textbf{Operation} & \textbf{AWS Nitro} & \textbf{Intel TDX} \\
\midrule
Wallet Creation & 6,531 & 14,970 \\
Account Balance Queries & 530,943 & 935,428 \\
Inbox-to-Subaccount Claims & 14,856 & 33,064 \\
Internal Subaccount Transfers & 32,953 & 71,292 \\
Withdrawal Request Creation & 18,291 & 39,702 \\
Transaction History Queries & 3,757 & 7,685 \\
End-to-End Workflow & 35,663 & 78,433 \\
\midrule
Batch Deposit \& Claim & 481 & 724 \\
Multi-threaded Operations (5 threads) & 32,189 & 55,124 \\
\bottomrule
\end{tabular}
\end{table*}

\textbf{Analysis.} The observed performance differences primarily stem from their distinct hardware architectures. AWS Nitro Enclaves employ a disaggregated system design, where the host server and the TEE hardware are physically separated. Communication occurs via vsock, and the enclave itself runs on a dedicated, lower-performance coprocessor. In contrast, Intel TDX leverages a special instruction set to provide enclave functionality directly on the main processor, avoiding the overhead of data transfer. As a result, Intel TDX is able to deliver significantly higher throughput for in-enclave operations compared to AWS Nitro Enclaves.

These preliminary benchmarks show that both Nitro and TDX backends can sustain high throughput for critical wallet operations, confirming that \sysname is practical for real-world deployment.

\textbf{Comparison with Related Work.}
To contextualize our performance results, we compare \sysname with other TEE-based blockchain systems and traditional custodial solutions.
Most existing wallet solutions operate under entirely different trust and security models that preclude meaningful performance comparisons. Hardware wallets prioritize air-gapped security over throughput, while custodial services sacrifice transparency for performance optimizations that \sysname cannot adopt without compromising its core security guarantees. Additionally, no prior work combines TEE-based execution with comprehensive provenance tracking and rule-based access controls, making \sysname's architecture fundamentally distinct from existing approaches.

\textit{Liquefaction.} The Liquefaction system~\cite{Austgen2024Liquefaction} represents the closest architectural parallel to \sysname, yet demonstrates significantly higher latencies. Liquefaction exhibits up to 1,036.9 seconds for finalized confirmation due to its dependence on external blockchain oracles for operation authorization. This design choice fundamentally limits scalability and user experience. In contrast, \sysname's innovative Inbox–Outbox architecture enables all internal operations to execute within the TEE without blockchain dependencies, achieving sub-millisecond latencies while maintaining equivalent security guarantees through cryptographic auditability.

\textit{Traditional Custodians.} While custodial solutions do not publish detailed performance benchmarks, they achieve high throughput by centralizing trust and sacrificing transparency—precisely the problems \sysname was designed to solve. Our TEE-based approach introduces only 5-7\% performance overhead compared to native execution~\cite{tdxperformance}, yet delivers fundamentally superior security properties: cryptographic auditability, decentralized operation, elimination of single points of failure, and user-controlled access policies. \sysname thus achieves the performance characteristics of centralized systems while providing the security guarantees traditionally associated only with self-custody solutions.

%% file: Sections/applications.tex
\section{Applications}\label{sec:applications}

\input{Sections/Figures/pass-applications}

\sysname is a powerful primitive for applications across three major blockchain
verticals: access control, custody \& compliance, and scalability.
Each vertical captures a distinct research area, while provenance-based
subaccounts provide a unifying abstraction.

\textbf{Access Control.} \sysname's multi-user subaccount model enables flexible permission scopes and safe delegation. As AI agents begin managing crypto assets, there is a clear need to constrain their authority. With \sysname, an agent can be granted a subaccount holding only limited funds or specific signing rights, allowing it to operate autonomously without risking the entire wallet; if compromised, damage is capped to that subaccount. Similarly, DAO voting can be delegated by granting signature rights restricted to governance DApps, without granting spending authority. Organizations can also provision \emph{identity subaccounts} that map to individuals or teams under a shared ENS or corporate identity, each with scoped balances and signing authority tracked through provenance.

\textbf{Custody \& Compliance.} Internal privacy, built-in permissioning, and auditable provenance make \sysname well-suited to enterprise and regulatory contexts. A flagship application is \emph{private payroll}: on-chain salary payments today risk leaking sensitive information such as compensation levels or departmental budgets. With \sysname, salaries are distributed privately as internal subaccount transfers, leaving a verifiable provenance trail for audits, but without exposing per-employee details on-chain. Only periodic batched payouts exit via the Outbox, balancing privacy with regulatory accountability. Beyond payroll, compliant organizational wallets can allocate explicit departmental budgets while enabling regulator-grade proofs derived from the provenance log, and B2B netting allows vendors to settle obligations internally, with only net amounts leaving the wallet to reduce leakage while maintaining auditability.

\textbf{Scalability.} Finally, \sysname supports blockchain scalability use cases where per-transaction gas costs or cross-domain coordination overhead are prohibitive. A flagship example is \emph{gasless appchains}: high-throughput applications such as central limit order books (CLOBs) can internalize matching and settlement as subaccount transitions, with the Outbox emitting only periodic, verifiable batch settlements. Similarly, cross-chain bridges can leverage the Inbox–Outbox abstraction to enforce provenance across domains and ensure FIFO nonce discipline to prevent replay, while L2 composability is achieved by treating each rollup or sidechain as a scoped subaccount, allowing private balance transfers internally and settling externally only on domain exits.

%% file: Sections/Figures/pass-applications.tex
\begin{table*}[h!]
  \renewcommand{\arraystretch}{1.3}
  \centering
  \scriptsize
  \begingroup
    \setlength{\tabcolsep}{4pt}
    \begin{NiceTabular}{%
      >{\columncolor{white}}p{1.4cm}
      >{\centering\arraybackslash}p{1.9cm}
      >{\raggedright\arraybackslash}p{2.8cm}
      >{\raggedright\arraybackslash}p{4.3cm}
    }[color-inside]
      \CodeBefore
        \rowcolors{2}{gray!20}{}
      \Body
      \toprule
      \rowcolor{gray!50}
      \textbf{Vertical} & \textbf{Application} & \textbf{Context} & \textbf{PASS Enables} \\
      \midrule

      \multirow{3}{*}{\shortstack{Access\\Control}}
        & \textbf{AI Agent Wallets}
          & Autonomous agents interacting with crypto assets~\cite{walters2025eliza} require bounded signing authority.
          & Scoped subaccounts grant limited balances and domain-specific rights, enabling safe delegation without exposing the root key. \\

        & DAO Voting Delegation
          & DAOs depend on address-based voting and proxy delegation~\cite{ethereum-daos}.
          & Fine-grained voting rights delegation auditable via provenance; delegation can be restricted to specific DApps or domains. \\

        & Identity Subaccounts
          & Organizations map ENS or organizational identities to individuals or teams.
          & Each subaccount holds scoped balances and signing rights, with provenance tracking under a shared organizational identity. \\
      \midrule

    \multirow{3}{*}{\makecell{Custody \&\\Compliance}}
        & \textbf{Private Payroll}
          & On-chain payrolls reveal sensitive salary information~\cite{auer2023bankingbis}.
          & Salary distribution through internal gasless transfers with auditable logs; external payouts are batched via the Outbox. \\

        & Compliant Org Wallets
          & Firms must maintain accountability and BSA/AML recordkeeping~\cite{fincen2020ThresholdRule}.
          & Departmental and individual subaccounts with explicit budgets; provenance ensures regulator-grade auditability and enables external proof exploration without disclosing private flows. \\

        & Private B2B Netting
          & Supply-chain and recurring multi-party payments require confidentiality.
          & Vendor subaccounts allow internal netting of obligations, with \sysname acting as a clearing house. Only net settlements exit via the Outbox, minimizing leakage while retaining auditability. \\
      \midrule

      \multirow{3}{*}{Scalability}
        & \textbf{Gasless Appchains}
          & High-throughput applications (e.g., trading engines) face prohibitive gas costs on-chain.
          & Internalize application logic as subaccount transitions; Outbox emits verifiable batch settlements, ensuring correctness with minimal gas. \\

        & Cross-chain Bridges
          & Bridges require strong custody and coordination of operators~\cite{wormhole_guardians}.
          & Inbox–Outbox abstractions enforce provenance across domains; FIFO nonce handling prevents replay or race conditions. \\

        & L2 Composability
          & Rollups and L2s fragment assets and user experience~\cite{rollupsurvey2022}.
          & Treat each L2 as a scoped subaccount; maintain private balance transfers internally, and settle externally only when crossing domains. \\

      \bottomrule
    \end{NiceTabular}
  \endgroup
  \caption{Applications of PASS across three blockchain verticals: access control, custody \& compliance, and scalability. Each vertical highlights flagship use cases enabled by provenance-based subaccounts.}
  \label{tab:pass-applications}
\end{table*}

%% file: Sections/future.tex
\section{Discussion and Future Work}\label{sec:future-work}
As \sysname matures beyond a single-wallet instantiation, several directions emerge for strengthening its composability, verifiability, and practical performance. Appendix~\ref{sec:design-extensions} provides a high-level construction of extension designs for \sysname.

\textbf{Composability and Inter\sysname Transactions.} Beyond single-wallet privacy and provenance, a key challenge is enabling secure interaction across wallets in supply-chain or inter-organizational settings. The trivial method of direct EOA transfers compromises privacy, so we envision each \sysname instance as a PCN (Payment Channel Network) hub, similar to XTransfer \cite{aumayr2025xtransfer}, with a central Hub aggregating signed payment intents $\mathcal{T} := {(a_i, x_i, s_i, r_i)}_{i=1}^n$, computing a netted settlement $\mathcal{S}$, and updating balances via a multi-party protocol akin to Thora \cite{cryptoeprint:thora2022}. Optimistically this remains off-chain and private, while disputes invoke atomic on-chain enforcement. This design leverages provenance guarantees for confidential, auditable inter-wallet transfers; full specification is left to future work.


\textbf{Instantiation on zkVM.} The current \sysname design primarily relies on trusted execution environments (TEEs) to safeguard keys and enforce Inbox and Outbox provenance rules, but this trust model depends on vendor attestation and hardware security, which may be insufficient in adversarial or regulatory contexts that demand independent verification. zkVMs provide a natural next step by enabling provenance-based control to be verified through succinct cryptographic proofs rather than trusted hardware. By batching transitions over $\mathcal{L}$, $\mathcal{H}$, $\mathcal{I}$, and $\mathcal{O}$ and producing a proof $\pi$, a minimalist verifier contract can update commitments and enforce outbox execution while preserving privacy. This direction complements the Chainless stack architecture~\cite{seong2025chainlessappsmodularframework}, suggesting a path toward modular, interoperable deployments of \sysname with stronger external verifiability and broader applicability.

\textbf{Expansion of Formal Model and Further Benchmarks.} We aim to extend our Lean 4 formalization from wallet logic to enclave assumptions, quorum dynamics, MPC-based recovery, and cross-vendor migration. Empirically, while alternatives (multisigs, ERC-4337, custodial systems) lack benchmarks, PASS already demonstrates throughput; we will expand evaluation to latency, concurrency, and vendor diversity. Together these efforts provide application-driven validation and stronger evidence for PASS as both formally sound and practically deployable.


%% file: Sections/conclusion.tex
\section{Conclusion}\label{sec:conclusion}


We presented \sysname, a Provenanced Access Subaccount System, as a new wallet architecture that enforces access by provenance rather than by identity or role. The Inbox–Outbox mechanism ensures all assets have a verifiable lineage, while internal transfers remain private and indistinguishable from a standard EOA account.  

Our contributions are threefold: (i) a formal model in Lean~4 proving privacy, accessibility, and provenance integrity; (ii) a working prototype with enclave backends on AWS Nitro Enclaves and Intel TDX via dstack, integrated with WalletConnect; and (iii) initial benchmarks showing the feasibility of high-throughput operation. These results demonstrate that provenance-based control can be realized today with mainstream blockchain tools. 

Looking forward, richer evaluations across multi-vendor TEEs, integration with zkVMs, and exploration of organizational and agent-based custody are promising directions. By combining rigorous formal guarantees with deployable prototypes, \sysname advances the design space between strict self-custody and flexible shared access, offering a path toward practical, privacy-preserving wallet security.

%% file: Sections/A-formal-verification.tex
\section{Formal Verification} \label{sec:formal-verification-extended}

We consider the formal model introduced in Section~\ref{sec:design} of \sysname as a state machine in Lean4. Our goal is to prove that, under the \sysname design, several critical security properties always hold. We outline these properties and our approach to verifying them.


\subsection{Property: Privacy of Internal Transfers}
\label{subsec:prop-privacy}

In the \sysname design, the \textsf{InternalTransfer} operation updates only the off-chain ledger $\mathcal{L}$ and the provenance history $\mathcal{H}$. It never modifies the Outbox $\mathcal{O}$, meaning there is no on-chain transaction triggered by an internal transfer. Consequently, external observers see no change in on-chain state. We formalize this by considering the on-chain view that an external party has for a given \sysname Wallet. We can model this visible on-chain state as a triple:
\[
{\mathcal{W}_{pub}} \;=\; \bigl(\pk,\; O,\; A\bigr),
\]
where:
\begin{itemize}[label={}, noitemsep, leftmargin=1.5em]
\item $\pk$ is the externally owned account (EOA) address,
\item $O = (T, \nu)$ is the external view of the \emph{outbox}, where $T$ is a list of enqueued transactions and $\nu$ is a nonce,
\item $A$ is the external view of the Asset ledger $\mathcal{L}$, containing a list of asset entries, each of the form $(\alpha, x)$ with a unique identifier $\alpha$ and total balance $x$.
\end{itemize}

\paragraph{External Trace.}
The \emph{external trace} of a PassAccount state $\mathcal{W}_{pub}$ is the list of externally visible actions, obtained by mapping each transaction in the outbox $T$ to an on-chain action:
\[
\operatorname{externalTrace}(\mathcal{W}_{pub}) \;=\; \bigl\{\, \operatorname{outboxTx}(t) \;\mid\; t \in T \bigr\}.
\]
No internal transfers appear in $\operatorname{externalTrace}(\cdot)$.

\paragraph{States That Differ Only by Internal Transfers.}
We say that two PassAccount states $\mathcal{W}_1 = (\pk_1,O_1,A_1)$ and $\mathcal{W}_2 = (\pk_2,O_2,A_2)$ \emph{differ only by internal transfers} if:
\begin{enumerate}
\item $\pk_1 = \pk_2$,
\item $O_1 = O_2$ (i.e.\ both the list of outbox transactions and the nonce are the same),
\item For every asset identifier $\alpha$, the total balance of that asset in $A_1$ equals the total balance of the same asset in $A_2$.
\end{enumerate}
Letting $\operatorname{getAsset}(\mathcal{W},\alpha)$ return the asset $\alpha$ within state $\mathcal{W}$ or $\bot$, we require $\forall\,\alpha$:
\[
\operatorname{totalBalance}\bigl(\operatorname{getAsset}(\mathcal{W}_1,\alpha)\bigr)
=
\operatorname{totalBalance}\bigl(\operatorname{getAsset}(\mathcal{W}_2,\alpha)\bigr).
\]

\begin{theorem} [Privacy of Internal Transfers]
If two PassAccount states $\mathcal{W}_1$ and $\mathcal{W}_2$ differ only by internal transfers, then their externally visible traces are identical:
\[
\operatorname{externalTrace}(\mathcal{W}_1) \;=\; \operatorname{externalTrace}(\mathcal{W}_2).
\]
\end{theorem}

\begin{proof}
By assumption, $O_1 = O_2$. Hence both states have the same outbox transaction list $T$. Since
\[
\operatorname{externalTrace}(\mathcal{W}_i) \;=\; \operatorname{map}\bigl(\operatorname{outboxTx},\, T_i\bigr)
\quad \text{for }i\in\{1,2\},
\]
and $T_1 = T_2$, it follows that
$\operatorname{externalTrace}(\mathcal{W}_1)$ and $\operatorname{externalTrace}(\mathcal{W}_2)$ are identical.
\end{proof}

\noindent
Thus, any sequence of purely internal transfers, changing only the private ledger $\mathcal{L}$ in \sysname, is invisible to an external observer, preserving the privacy of how assets move within the wallet.

A corollary of this Internal Transfer property is that a \sysname account is indistinguishable from a regular EOA account.

\begin{corollary}[Indistinguishability from EOA]
Let $\mathcal{W}_{\mathrm{pass}}$ be a \sysname state that (i) has the same public key $\pk$, (ii) the same outbox, including nonce and transaction list, and (iii) the same total balances of each asset as a reference EOA state $\mathcal{W}_{\mathrm{eoa}}$. Then
\[
\operatorname{externalTrace}(\mathcal{W}_{\mathrm{pass}})
\;=\;
\operatorname{externalTrace}(\mathcal{W}_{\mathrm{eoa}}).
\]
In particular, an external observer cannot distinguish $\mathcal{W}_{\mathrm{pass}}$ from $\mathcal{W}_{\mathrm{eoa}}$.
\end{corollary}
\begin{proof}
By assumption, the only potential differences between $\mathcal{W}_{\mathrm{pass}}$ and $\mathcal{W}_{\mathrm{eoa}}$ lie in internal ledger details. From Theorem~1, such internal transfers do not affect the externally visible outbox transactions or total on-chain balances. Thus, their external traces coincide, and $\mathcal{W}_{\mathrm{pass}}$ is indistinguishable from $\mathcal{W}_{\mathrm{eoa}}$.
\end{proof}


\subsection{Property: Asset Accessibility}

In addition to keeping internal transactions private, \sysname should ensure that for every asset in the system, there exists some sub-account that can access it and withdraw it from the system. This prevents a situation where funds get "stuck" without any reachable key or policy to move them. We define asset accessibility as follows:

\begin{theorem}[Asset Accessibility]
For any \sysname account $\mathcal{W}$ and asset identified by $\alpha$, there exists a list of addresses $L$ such that:
\begin{enumerate}
    \item For every address $u \in L$, the user can access their full balance: 
    \[\mathtt{checkBalance}(s, \alpha, u, \mathtt{balance}(u)) = \mathtt{true}\]
    \item The sum of all accessible balances equals the total balance of the asset:
    \[\sum_{u \in L} \mathcal{L}[u][\alpha] = \mathtt{totalBalance}(\mathtt{\alpha})\]
\end{enumerate}
\end{theorem}

\begin{proof}
We construct $L$ as the set of all addresses with non-zero balances in the balance map for $\alpha$:
\[L = \{u \; | \; \mathcal{L}[u][\alpha] > 0\}\]

For the first property, given any $u \in L$, the definition of $\mathtt{checkBalance}$ verifies if the user has sufficient balance for the specified amount. When checking against the user's own balance, this is trivially satisfied since $\mathcal{L}[u][\alpha] \geq \mathcal{L}[u][\alpha]$ for all $u$.

For the second property, we observe that $\mathtt{totalBalance}$ is defined as the sum of all balances in the asset's balance map:
\[\mathtt{totalBalance}(\alpha) = \sum_{u \in \mathcal{L}} \mathcal{L}[u][\alpha] \]

By construction, $L$ contains the only nonzero values $u \in \mathcal{L}$, so:
\[\sum_{u \in L} \mathcal{L}[u][\alpha] = \mathtt{totalBalance}(\alpha)\]

Therefore, both properties hold for our constructed list $L$.
\end{proof}

We can prove a similar trait for General Signable Messages, showing that for any inbound GSM, there is some subaccount signer that has access to it.

\begin{theorem}[GSM Accessibility]
For any GSM and domain, there exists a user with signing access to that domain.
\end{theorem}

\begin{proof}
For any domain $d$, we can directly construct a user $u$ who has signing access to $d$ as follows:
\[u = \mathtt{getSigner}(d)\]

By definition, $\mathtt{getSigner}$ returns either:
\begin{itemize}
    \item A specifically assigned signer if one exists in the domain-address map, or
    \item The default signer otherwise
\end{itemize}

Therefore, every domain is guaranteed to have at least one user with signing access, ensuring no GSM domain becomes inaccessible.
\end{proof}

\subsection{Property: Provenance Integrity}

Provenance integrity specifies that every asset held by any sub-account in \sysname must have an origin traceable to an external deposit into the Inbox. At any state of the system, for each asset entry in a sub-account’s balance, we can find a sequence of internal transfers that leads back to an Inbox deposit from an external address. This property ensures no sub-account can conjure funds from nothing, such as through double-spending or inflation within the system, and that the internal ledger is consistent with on-chain reality. In Lean, we model the state of the \sysname wallet including Inbox, Outbox, and all sub-accounts and prove an invariant that whenever a sub-account has $x$ units of asset $\alpha$, the global state contains a record of at least $x$ units of $\alpha$ having been deposited, minus any that have been withdrawn.

\begin{theorem}[Provenance Integrity]
\label{thm:provenance-integrity}

For any \sysname account 
$S = \bigl(\pk, \sk, \nu, \mathcal{I}, \mathcal{O}, \mathcal{L}, \mathcal{H}\bigr)$
and any asset $\alpha$, the total amount of $\alpha$ held across all subaccounts can never exceed the total amount of external deposits (minus any that have already been withdrawn). Formally,
\[
\sum_{u \in U} \mathcal{L}[u][\alpha]
\;\;\le\;\;
\mathtt{externalDeposits}(\alpha),
\]
where $U$ is the set of subaccounts in the system, and $\mathtt{externalDeposits}(\alpha)$ is the net total of asset $\alpha$ that has been introduced into the wallet from outside addresses, recorded in the inbox $\mathcal{I}$.
\end{theorem}

\begin{proof}
We model the system’s state $(\pk, \sk, \nu, \mathcal{I}, \mathcal{O}, \mathcal{L}, \mathcal{H})$ as in Figure~\ref{fig:pass-symbolic-create}, including the Inbox $\mathcal{I}$, Outbox $\mathcal{O}$, and the internal ledger $\mathcal{L}$. We define a \emph{provenance record} $P(\alpha, u)$ that captures the total amount of asset $\alpha$ subaccount $u$ is authorized to hold, based on Inbox claims and internal transfers.

\paragraph{Base Case:}
Immediately after \textbf{CreatePassWallet}, no assets have been deposited, and no subaccount balances exist:
\[
\forall\,u,\; \mathcal{L}[u][\alpha] = 0, 
\quad
\forall\,\alpha,\;\mathtt{externalDeposits}(\alpha) = 0.
\]
Hence,
\[
\sum_{u \in U} \mathcal{L}[u][\alpha]
\;=\;
0
\;\;\le\;\;
0
\;=\;
\mathtt{externalDeposits}(\alpha).
\]

\paragraph{Inductive Step:}
Assume the invariant holds after $n-1$ operations. We prove it remains valid after the $n$th operation, for each possible type of operation:

\begin{enumerate}[label=(\arabic*), leftmargin=*]
\item \textbf{External Deposit:}  
  Suppose an amount $x$ of asset $\alpha$ is deposited into the wallet:
  \[
    \mathtt{externalDeposits}(\alpha)
    \;\gets\;
    \mathtt{externalDeposits}(\alpha) + x.
  \]
  No subaccount balances $\mathcal{L}[u][\alpha]$ change. Thus, the left side of our inequality (sum of subaccount balances) is unchanged, while the right side strictly increases by $x$, so the invariant continues to hold.

\item \textbf{Inbox Claim:}  
  If subaccount $u$ claims amount $x$ of asset $\alpha$ from the Inbox $\mathcal{I}$:
  \[
    \mathcal{L}[u][\alpha] 
    \;\gets\; 
    \mathcal{L}[u][\alpha] + x,
    \quad
    P(\alpha, u) 
    \;\gets\;
    P(\alpha, u) + x.
  \]
  Since $x$ was already accounted for in
  $\mathtt{externalDeposits}(\alpha)$ (but not yet allocated to any
  subaccount), the total subaccount balances remain bounded by
  $\mathtt{externalDeposits}(\alpha)$.

\item \textbf{Internal Transfer:}  
  If subaccount $u_1$ transfers $x$ of $\alpha$ to subaccount $u_2$:
  \[
    \mathcal{L}[u_1][\alpha] 
    \;\gets\;
    \mathcal{L}[u_1][\alpha] - x,
    \quad
    \mathcal{L}[u_2][\alpha] 
    \;\gets\;
    \mathcal{L}[u_2][\alpha] + x,
  \]
  and $P(\alpha, u_2)$ increases by $x$. Because the sum of all subaccount balances does not change, our invariant is preserved.

\item \textbf{Withdrawal:}  
  If subaccount $u$ withdraws $x$ of $\alpha$:
  \[
    \mathcal{L}[u][\alpha] 
    \;\gets\;
    \mathcal{L}[u][\alpha] - x,
  \]
  which reduces the total subaccount balances. Since $\mathtt{externalDeposits}(\alpha)$ remains unchanged (assets have simply exited the wallet), the inequality still holds.
\end{enumerate}

Thus, by induction over all possible operations, 
\[
\sum_{u \in U} \mathcal{L}[u][\alpha]
\;\;\le\;\;
\mathtt{externalDeposits}(\alpha)
\]
for every reachable state. No subaccount can create new assets out of nothing, and all assets in the system are fully accounted for by provenanced deposits.
\end{proof}


%% file: Sections/A-dstack-system.tex
\section{Dstack Architecture and Execution Model}\label{sec:dstack}

\subsection{Component Overview}
\begin{itemize}
    \item \textbf{Dstack-OS:} Minimal Linux-based VM image with a verifiable boot chain. Measures bootloader, kernel, root filesystem (RootFs) via MRTD and RTMR registers. Integrates \textsf{dm-verity} for data integrity and monotonic counters for anti-rollback.
    \item \textbf{Dstack-KMS:} Blockchain-governed MPC network running in TEE nodes for key derivation and rotation. Provides:
    \begin{align*}
        k_c &= \textsf{HMAC}_{k_{\mathrm{root}}}(c_{\mathrm{id}})
    \end{align*}
    Rotation is triggered via governance ($t$-of-$n$ threshold).
    \item \textbf{Dstack-Gateway/Ingress:} TLS termination inside TEEs.
    Its ZeroTrust-TLS protocol binds $\mathcal{X}_c$ to $\mathcal{R}(c)$,
    enabling verifiable HTTPS ingress.
\end{itemize}

\subsection{Governance Layer}
Two contracts form the trust base:
\begin{itemize}
    \item \emph{DstackKms}: Global registry of KMS nodes and application digests.
    \item \emph{DstackApp}: Per-application contract specifying quorum rules, allowed image hashes, and authorized enclaves.
\end{itemize}

\subsection{Workflow Phases}
\begin{enumerate}
    \item \textbf{Attestation:} Enclave $e$ produces quote $Q_e$; verified against $\mathcal{R}$.
    \item \textbf{Secret Provisioning:} MPC nodes compute $k_c$, release to $e$ if $Q_e$ is valid.
    \item \textbf{Container Launch:} Dstack-OS mounts RootFs, decrypts state $\sigma_c$ with $k_c$.
    \item \textbf{Ingress Binding:} TLS certificate $\mathcal{X}_c$ generated inside $e$; bound to $\mathcal{R}(c)$.
    \item \textbf{Audit Logging:} Migration, key rotation, and governance events appended to $\mathcal{H}$.
\end{enumerate}

\subsection{Fault Tolerance}
\begin{itemize}
    \item Vendor failure: $\mu$ enables re-deployment across Intel/AMD TEEs.
    \item KMS compromise: MPC-threshold key rotation restores
    forward and backward secrecy.
    \item Governance deadlock: timeout fallback in $\Gamma$
    ensures quorum reduction and forward progress.
\end{itemize}

\subsection{Liveness Proof Sketch}
Combining MPC-based key recovery, quorum fallback, and portable containers yields:
\[
\forall \textsf{op}\in\mathcal{O}_c,\quad \Pr[\textsf{Exec}(c,\textsf{op})=\top]=1
\]
under assumptions of $f$-resilient MPC nodes and at least one uncompromised TEE KMS instance.

%% file: Sections/A-benchmark-results.tex
\section{Detailed Benchmark Results}\label{appendix:benchmark-results}

This appendix provides comprehensive benchmark statistics for both AWS Nitro Enclaves and Intel TDX environments. All operations are internal \sysname state changes that occur within the TEE with no blockchain network interaction. Each operation was executed 100 times per trial, with 10 independent trials per environment. The tables below show summary statistics including mean, standard deviation, and range.

\subsection{AWS Nitro Enclaves Results}

\begin{table}[H]
\centering
\caption{AWS Nitro Enclaves benchmark statistics (ops/sec) across 10 trials}
\label{tab:nitro-detailed}
\small
\begin{tabular}{lrrrr}
\toprule
\textbf{Operation} & \textbf{Mean} & \textbf{Std Dev} & \textbf{Min} & \textbf{Max} \\
\midrule
Wallet Creation & 6,531 & 531 & 5,446 & 7,042 \\
Account Balance Queries & 530,943 & 675 & 530,216 & 531,915 \\
Inbox-to-Subaccount Claims & 14,856 & 476 & 14,326 & 15,504 \\
Internal Subaccount Transfers & 32,953 & 1,322 & 31,447 & 35,088 \\
Withdrawal Request Creation & 18,291 & 1,254 & 15,873 & 19,608 \\
Transaction History Queries & 3,757 & 45 & 3,650 & 3,802 \\
End-to-End Workflow & 35,663 & 6,234 & 21,739 & 40,650 \\
\midrule
Batch Deposit \& Claim & 481 & 13 & 467 & 497 \\
Multi-threaded Operations (5 threads) & 32,189 & 884 & 31,200 & 34,364 \\
\bottomrule
\end{tabular}
\end{table}

\subsection{Intel TDX Results}

\begin{table}[H]
\centering
\caption{Intel TDX benchmark statistics (ops/sec) across 10 trials}
\label{tab:dstack-detailed}
\small
\begin{tabular}{lrrrr}
\toprule
\textbf{Operation} & \textbf{Mean} & \textbf{Std Dev} & \textbf{Min} & \textbf{Max} \\
\midrule
Wallet Creation & 14,970 & 4,652 & 8,696 & 24,390 \\
Account Balance Queries & 935,428 & 338,450 & 383,142 & 1,470,588 \\
Inbox-to-Subaccount Claims & 33,064 & 5,197 & 25,445 & 41,667 \\
Internal Subaccount Transfers & 71,292 & 17,623 & 47,170 & 101,010 \\
Withdrawal Request Creation & 39,702 & 7,434 & 29,155 & 51,813 \\
Transaction History Queries & 7,685 & 872 & 6,452 & 8,929 \\
End-to-End Workflow & 78,433 & 17,403 & 59,524 & 110,497 \\
\midrule
Batch Deposit \& Claim & 724 & 81 & 636 & 893 \\
Multi-threaded Operations (5 threads) & 55,124 & 16,744 & 28,329 & 81,301 \\
\bottomrule
\end{tabular}
\end{table}

\subsection{Performance Analysis}

The statistical analysis reveals several key insights about \sysname's internal operation performance:

\begin{itemize}
    \item \textbf{Performance Consistency}: AWS Nitro Enclaves demonstrate superior consistency with lower standard deviations across all operations. The coefficient of variation ranges from 0.1\% (Account Balance Queries) to 17.5\% (End-to-End Workflow).
    \item \textbf{Peak Performance}: Intel TDX achieves higher mean throughput but with significantly greater variance. Standard deviations are 10-500x larger than Nitro, indicating less predictable performance.
    \item \textbf{Production Considerations}: These internal operation benchmarks demonstrate \sysname's ability to handle thousands of wallet operations per second within the TEE, with the actual deployment bottleneck being blockchain transaction submission rather than internal processing.
\end{itemize}

The statistical summary confirms that both platforms support practical \sysname deployment, with Nitro offering predictable performance and Intel TDX providing higher peak throughput at the cost of consistency.

%% file: Sections/A-design-extensions.tex
\section{Design Extensions}\label{sec:design-extensions}

\subsection{Composability and Inter\sysname Transactions}

While \sysname provides privacy and provenance guarantees for a single wallet instance, 
an important question is how multiple \sysname wallets may interact. This is important in supply-chain use cases where each \sysname represents a single company in the supply chain. 

Trivially, \sysname A can send an on-chain transaction from its Outbox to \sysname B's EOA account, and this will land in the Inbox of \sysname B, ready to claim. However, this undermines the privacy guarantees that \sysname offers at the EOA level, and is undesirable in cases where inter-organizational asset transfers contain sensitive information, such as order details.

We propose a non-trivial approach inspired by cross-Payment Channel Network (PCN) protocols such as XTransfer \cite{aumayr2025xtransfer}. We can treat each \sysname instance akin to a PCN hub node, with state $\mathcal{L}$ and transaction history $\mathcal{H}$. As with XTransfer, we can assume a star topology of \sysname nodes, with a central \sysname Hub in communciation with all \sysname nodes. The \sysname Hub will aggregate all Inter\sysname transactions $\mathcal{T} := \{(a_i, x_i, s_i, r_i)\}_{i=1}^n$, where $(a_i, x_i, s_i, r_i)$ represent the asset, amount, sender, and receiver, each signed by the $s_i$ Outbox.

The Hub aggregates these signed payment intents, computes a netted settlement plan $\mathcal{S}$, and coordinates a secure multi-party update protocol to commit the resulting balances, similar to Thora \cite{cryptoeprint:thora2022}. In the optimistic case, this process completes entirely off-chain, and each participant’s balance is updated in the \sysname ledger without revealing the underlying transaction graph. If a dispute or failure occurs, any participant may trigger an on-chain enforcement transaction that finalizes $\mathcal{S}$ atomically via the direct transfer mechanism, thus guaranteeing safety.

This approach leverages \sysname’s provenance and state design to make each instance an atomic node in PCN-style systems like X-Transfer. Here, we outline how \sysname Hubs can enable confidential, auditable cross-organization transfers, while full formal specification and protocol implementation is deferred for future work.

\subsection{Instantiation on zkVM}
The design of \sysname presented thus far relies on trusted execution environments (TEEs) such as AWS Nitro or Intel TDX to enforce provenance-based control. TEEs provide strong isolation and efficient execution, allowing internal transfers to remain private while ensuring that every outbox action has a verifiable lineage. However, the trust model is anchored in hardware attestation: users and regulators must accept that the enclave executes the correct code as claimed. While sufficient for many organizational custody scenarios, this assumption limits transparency and long-term auditability. To address this limitation, we explore an alternative instantiation of \sysname in a zero-knowledge virtual machine (zkVM), replacing hardware trust with cryptographic proofs of correctness.

Zero-knowledge virtual machines (zkVMs) provide a proving environment in which arbitrary state machines can run with strong correctness guarantees~\cite{arun2023zkvm}. By producing a succinct proof $\pi$ of execution, zkVMs offer integrity assurances similar to TEEs but without reliance on vendor-specific hardware. This shift makes zkVMs particularly suitable for contexts such as institutional custody, multi-party supply chains, or inter-organizational accounting where external verifiability and regulator-grade auditability are paramount.

In a zkVM instantiation of \sysname, the execution pipeline mirrors the enclave-based design but replaces attestation with cryptographic proofs. The process proceeds in four stages:
\begin{enumerate}
    \item \textbf{Batch formation.} Transaction requests (claims, transfers, withdrawals) are first queued off-chain into a batch.
    \item \textbf{zkVM execution.} The zkVM program takes as private inputs: (i) the prior batch’s commitments of the ledger $\mathcal{L}$ and provenance log $\mathcal{H}$, (ii) the queued transitions, including inbox $\mathcal{I}$ deposits and outbox $\mathcal{O}$ withdrawals, and (iii) the proposed new commitments $(\mathcal{L}', \mathcal{H}')$. Inside the zkVM, the transition function is replayed deterministically: starting from the previous roots, applying all operations in sequence, updating balances and provenance entries, and checking consistency constraints.
    \item \textbf{Proof generation.} The zkVM outputs a succinct proof $\pi$ certifying that the claimed new roots $(\mathcal{L}', \mathcal{H}')$ are correct with respect to the given transition set.
    \item \textbf{On-chain verification.} A minimalist verifier contract validates $\pi$, updates the recorded commitments, and enforces any pending outbox operations. This ensures that every external action has a verifiable provenance path while internal transfers remain private.
\end{enumerate}

Compared to TEEs, zkVMs strengthen the verifiability model by eliminating hardware trust assumptions, but incur higher computational and financial costs due to proof generation and on-chain verification. Batching amortizes these costs, raising open research questions about optimal batch sizes, data availability strategies, and proof system selection. zkVM backends such as Risc0, SP1, or Plonky3 offer different tradeoffs in proof size, verification cost, and latency. Hybrid deployments are also possible: TEEs can serve as fast-path executors for intra-day operations, with zkVM proofs generated periodically to provide cryptographic audit trails. This hybrid pattern resembles optimistic rollups, where fast execution is combined with verifiable settlement.

Finally, this instantiation aligns closely with the Chainless stack architecture~\cite{seong2025chainlessappsmodularframework}, which emphasizes the separation of execution, trust, bridging, and settlement. In this view, zkVM proofs naturally inhabit the execution and trust layers, while Inbox–Outbox commitments integrate with bridging and the verifier contract finalizes settlement. This layered design opens the door to modular, interoperable deployments of \sysname across heterogeneous chains and domains. Full formal specification, performance benchmarking, and exploration of hybrid zkVM--TEE implementations remain important avenues for future work.

\begin{figure}[H]
  \centering
  \includegraphics[width=0.5\linewidth]{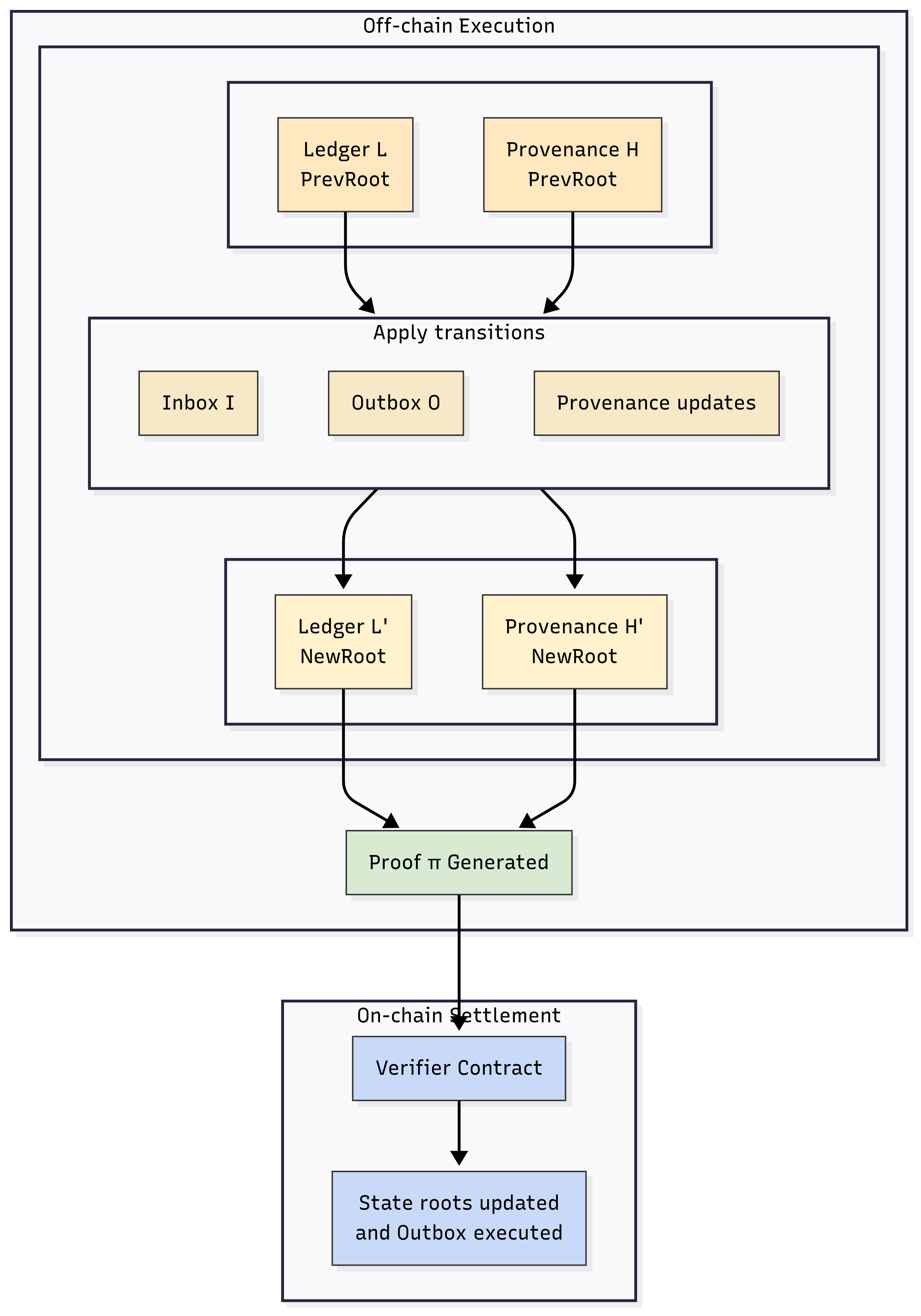}
  \caption{PASS instantiation inside a zkVM: Inbox $\mathcal{I}$ and Outbox $\mathcal{O}$ transitions are applied to the Ledger $\mathcal{L}$ and Provenance $\mathcal{H}$, producing new state commitments. The zkVM outputs a proof $\pi$, which the on-chain verifier checks before updating roots and enforcing outbox operations.}
  \label{fig:pass-zkvm}
\end{figure}

%% file: references.bib
@misc{tdxperformance,
  title={{Performance Considerations of Intel® Trust Domain Extensions on 4th Generation Intel® Xeon® Scalable Processors}},
  author={Intel},
  year={2025},
  url={https://www.intel.com/content/www/us/en/developer/articles/technical/trust-domain-extensions-on-4th-gen-xeon-processors.html}
}

@misc{zhou2025dstackzerotrustframework,
      title={{Dstack: A Zero Trust Framework for Confidential Containers}}, 
      author={Shunfan Zhou and Kevin Wang and Hang Yin},
      year={2025},
      eprint={2509.11555},
      archivePrefix={arXiv},
      primaryClass={cs.CR},
      url={https://arxiv.org/abs/2509.11555}, 
}

@article{walters2025eliza,
  title     = {Eliza: A Web3 Friendly AI Agent Operating System},
  author    = {Shaw Walters and Sam Gao and Shakker Nerd and Feng Da and Warren Williams and Ting-Chien Meng and Hunter Han and Frank He and Allen Zhang and Ming Wu and Timothy Shen and Maxwell Hu and Jerry Yan},
  journal   = {arXiv preprint arXiv:2501.06781},
  year      = {2025},
  url       = {https://arxiv.org/pdf/2501.06781v1}
}

@manual{lean_documentation,
  title     = {Lean Documentation},
  author    = {{Lean Community}},
  year      = {2025},
  url       = {https://lean-lang.org/documentation/}
}

@manual{siwe_eip4361,
  title     = {EIP-4361: Sign-In with Ethereum},
  author    = {{Sign-In with Ethereum Documentation}},
  year      = {2023},
  url       = {https://docs.login.xyz/general-information/siwe-overview/eip-4361}
}

@manual{snapshot_voting,
  title     = {Voting on Snapshot},
  author    = {{Snapshot Documentation}},
  year      = {2023},
  url       = {https://docs.snapshot.box/user-guides/voting/vote}
}

@manual{erc4337_user_operation,
  title     = {UserOperation – ERC-4337 Documentation},
  author    = {{ERC-4337 Documentation}},
  year      = {2023},
  url       = {https://www.erc4337.io/docs/understanding-ERC-4337/user-operation}
}

@misc{pertsev2019tornado,
  title     = {Tornado Cash Privacy Solution Version 1.4},
  author    = {Alexey Pertsev and Roman Semenov and Roman Storm},
  year      = {2019},
  url       = {https://berkeley-defi.github.io/assets/material/Tornado%20Cash%20Whitepaper.pdf}
}

@article{safe2024eip7702,
  title     = {EIP-7702: A Win for Smart Accounts in Ethereum’s Pectra Upgrade?},
  author    = {Valentin Seehausen},
  journal   = {Safe Global Blog},
  year      = {2024},
  url       = {https://safe.global/blog/eip-7702-smart-accounts-ethereum-pectra-upgrade}
}

@manual{venly_omnibus_vs_segregated_wallets,
  title     = {Omnibus vs Segregated Wallets in Crypto},
  author    = {{Venly}},
  year      = {2025},
  url       = {https://docs.venly.io/docs/omnibus-vs-segregated-wallets-in-crypto}
}

@article{cointelegraph2024smart,
  title     = {What Are Smart Contract Wallets?},
  author    = {Dilip Kumar Patairya},
  journal   = {Cointelegraph},
  year      = {2024},
  url       = {https://cointelegraph.com/explained/what-are-smart-contract-wallets}
}

@manual{argent_docs,
  title     = {Build with Argent},
  author    = {{Argent}},
  year      = {2023},
  url       = {https://docs.argent.xyz/}
}

@manual{loopring_wallet_docs,
  title     = {Loopring Smart Wallet Documentation},
  author    = {{Loopring}},
  year      = {2023},
  url       = {https://docs-wallet.loopring.io/}
}

@manual{tally_gnosis_safe,
  title     = {Gnosis Safe Overview},
  author    = {{Tally}},
  year      = {2025},
  url       = {https://docs.tally.xyz/set-up-and-technical-documentation/using-governor-with-gnosis-safe/gnosis-safe}
}

@article{forbes_bybit_hack,
  title     = {Breaking: Could {Bybit}'s \$1.5{B} Hack Have Been Stopped? {Ledger}, {CZ} React},
  author    = {Nina Bambysheva},
  journal   = {Forbes},
  year      = {2025},
  url       = {https://www.forbes.com/sites/digital-assets/2025/02/22/breaking-could-bybits-14b-hack-have-been-stopped-ledger-cz-react/}
}

@manual{wormhole_guardians,
  title     = {Guardians},
  author    = {{Wormhole Foundation}},
  year      = {2025},
  url       = {https://wormhole.com/docs/learn/infrastructure/guardians/}
}

@inproceedings{mangipudi2023uncovering,
  title     = {Uncovering Impact of Mental Models towards Adoption of Multi-device Crypto-Wallets},
  author    = {Easwar Vivek Mangipudi and Umang Desai and Mohammad Minaei and Mainack Mondal and Aniket Kate},
  booktitle = {Proceedings of the 2023 ACM SIGSAC Conference on Computer and Communications Security},
  pages     = {3153--3167},
  year      = {2023},
  doi       = {10.1145/3576915.3623218},
  url       = {https://dl.acm.org/doi/10.1145/3576915.3623218}
}

@article{goldmansachs2021crypto,
  title     = {Crypto: A New Asset Class?},
  author    = {{Goldman Sachs Research}},
  year      = {2021},
  url       = {https://www.goldmansachs.com/insights/top-of-mind/crypto-a-new-asset-class}
}

@ARTICLE{rollupsurvey2022,
  author={Thibault, Louis Tremblay and Sarry, Tom and Hafid, Abdelhakim Senhaji},
  journal={IEEE Access}, 
  title={Blockchain Scaling Using Rollups: A Comprehensive Survey}, 
  year={2022},
  volume={10},
  number={},
  pages={93039-93054},
  doi={10.1109/ACCESS.2022.3200051}}

@manual{safe_docs_multisig,
    title = {{How do Safe Smart Accounts work?}},
    author = {{Safe}},
    year = {2025},
    url = {https://docs.safe.global/advanced/smart-account-concepts}
}

@misc{fireblocks_governance_policy_engine,
  title     = {Governance and Policy Engine},
  author    = {{Fireblocks}},
  year      = {2025},
  url       = {https://www.fireblocks.com/platforms/governance-and-policy-engine/}
}

@misc{turnkey_policy_quickstart,
  title     = {Policy Quickstart},
  author    = {{Turnkey}},
  year      = {2025},
  url       = {https://docs.turnkey.com/concepts/policies/quickstart}
}

@article{cutler2024cedar,
  title     = {Cedar: A New Language for Expressive, Fast, Safe, and Analyzable Authorization},
  author    = {Joseph W. Cutler and Craig Disselkoen and Aaron Eline and Shaobo He and Kyle Headley and Michael Hicks and Kesha Hietala and Eleftherios Ioannidis and John Kastner and Anwar Mamat and Darin McAdams and Matt McCutchen and Neha Rungta and Emina Torlak and Andrew M. Wells},
  journal   = {arXiv preprint arXiv:2403.04651},
  year      = {2024},
  url       = {https://arxiv.org/pdf/2403.04651}
}

@article{Austgen2024Liquefaction,
  author    = {James Austgen and Andr{\'e}s F{\'a}brega and Mahimna Kelkar and Dani Vilardell and Sarah Allen and Kushal Babel and Jay Yu and Ari Juels},
  title     = {{Liquefaction}: Privately Liquefying Blockchain Assets},
  journal   = {arXiv e-prints},
  volume    = {2412.02634},
  year      = {2024},
  note      = {arXiv:2412.02634 [cs.CR]}
}

@misc{oasis_protocol,
  author       = {{Oasis Labs}},
  title        = {Oasis Protocol: Privacy-Enabled Blockchain Platform},
  year         = {2024},
  url          = {https://oasisprotocol.org/},
  note         = {Accessed: 2024-10-06},
  howpublished = {\url{https://oasisprotocol.org/}}
}

@misc{schneider2022sokhardwaresupportedtrustedexecution,
      title={SoK: Hardware-supported Trusted Execution Environments}, 
      author={Moritz Schneider and Ramya Jayaram Masti and Shweta Shinde and Srdjan Capkun and Ronald Perez},
      year={2022},
      eprint={2205.12742},
      archivePrefix={arXiv},
      primaryClass={cs.CR},
      url={https://arxiv.org/abs/2205.12742}, 
}

@inproceedings{mckeen2013innovative,
	title={Innovative instructions and software model for isolated execution.},
	author={McKeen, Frank and Alexandrovich, Ilya and Berenzon, Alex and Rozas, Carlos V and Shafi, Hisham and Shanbhogue, Vedvyas and Savagaonkar, Uday R},
	booktitle={HASP},
	year={2013},
	pages={10}
}

@misc{ethereum-daos,
  author = {Ethereum.org Contributors},
  title = {Decentralized autonomous organizations ({DAO}s)},
  howpublished = {\url{https://ethereum.org/en/dao/}},
  note = {Accessed: 2024-11-13}
}

@misc{buterin2014ethereum,
  author = {Vitalik Buterin},
  title = {Ethereum: A Next-Generation Smart Contract and Decentralized Application Platform},
  year = {2014},
  howpublished = {Ethereum whitepaper, \url{https://ethereum.org/en/whitepaper/}},
  note = {Accessed: 2024-10-16}
}

@misc{aumayr2025xtransfer,
      author = {Lukas Aumayr and Zeta Avarikioti and Iosif Salem and Stefan Schmid and Michelle Yeo},
      title = {X-Transfer: Enabling and Optimizing Cross-{PCN} Transactions},
      howpublished = {Cryptology {ePrint} Archive, Paper 2025/272},
      year = {2025},
      url = {https://eprint.iacr.org/2025/272}
}

@misc{cryptoeprint:thora2022,
      author = {Lukas Aumayr and Kasra Abbaszadeh and Matteo Maffei},
      title = {Thora: Atomic and Privacy-Preserving Multi-Channel Updates},
      howpublished = {Cryptology {ePrint} Archive, Paper 2022/317},
      year = {2022},
      doi = {10.1145/3548606.3560556},
      url = {https://eprint.iacr.org/2022/317}
}

@misc{seong2025chainlessappsmodularframework,
      title={Chainless Apps: A Modular Framework for Building Apps with Web2 Capability and Web3 Trust}, 
      author={Brian Seong and Paul Gebheim},
      year={2025},
      eprint={2505.22989},
      archivePrefix={arXiv},
      primaryClass={cs.CR},
      url={https://arxiv.org/abs/2505.22989}, 
}

@misc{arun2023zkvm,
      author = {Arasu Arun and Srinath Setty and Justin Thaler},
      title = {Jolt: {SNARKs} for Virtual Machines via Lookups},
      howpublished = {Cryptology {ePrint} Archive, Paper 2023/1217},
      year = {2023},
      url = {https://eprint.iacr.org/2023/1217}
}

@misc{fincen2020ThresholdRule,
  title        = {Threshold for the Requirement To Collect, Retain, and Transmit Information on Funds Transfers and Transmittals of Funds That Begin or End Outside the United States, and Clarification of the Requirement To Collect, Retain, and Transmit Information on Transactions Involving Convertible Virtual Currencies and Digital Assets With Legal Tender Status},
  author       = {{Financial Crimes Enforcement Network} and {Department of the Treasury}},
  institution  = {Federal Register},
  year         = {2020},
  month        = oct,
  day          = {27},
  volume       = {85},
  number       = {20/23756},
  pages        = {68005--68019},
  address      = {Washington, D.C.},
  note         = {Proposed Rule, Docket No. FINCEN-2020-0002; RIN 1506-AB41, 31 CFR Parts 1010, 1020},
  url          = {https://www.federalregister.gov/documents/2020/10/27/2020-23756/threshold-for-the-requirement-to-collect-retain-and-transmit-information-on-funds-transfers-and}
}

@misc{helios_lightclient,
  author       = {{a16z crypto}},
  title        = {Helios: Ethereum Light Client},
  howpublished = {\url{https://github.com/a16z/helios}},
  note         = {GitHub repository, accessed: 2025-09-14},
  year         = {2025}
}

@misc{patlan2025aiagentscryptoland,
      author = {Atharv Singh Patlan and Peiyao Sheng and S. Ashwin Hebbar and Prateek Mittal and Pramod Viswanath},
      title = {{AI} Agents in Cryptoland: Practical Attacks and No Silver Bullet},
      howpublished = {Cryptology {ePrint} Archive, Paper 2025/526},
      year = {2025},
      url = {https://eprint.iacr.org/2025/526}
}

@techreport{auer2023bankingbis,
  author       = {Raphael Auer and Marc Farag and Ulf Lewrick and Lovrenc Orazem and Markus Zoss},
  title        = {Banking in the shadow of bitcoin? The institutional adoption of cryptocurrencies},
  institution  = {Bank for International Settlements},
  type         = {BIS Working Papers},
  number       = {1013},
  year         = {2023},
  url          = {https://www.bis.org/publ/work1013.htm}
}

@inproceedings{brunner2021keystorage,
  author       = {Brunner, Dominik and Karame, Ghassan O. and Li, Wenting and Tschorsch, Florian},
  title        = {Who Stores the Private Key? An Exploratory Study about User Preferences of Key Management for Blockchain-based Applications},
  booktitle    = {Proceedings of the 7th International Conference on Information Systems Security and Privacy (ICISSP)},
  pages        = {276--287},
  year         = {2021},
  publisher    = {SciTePress},
  doi          = {10.5220/0010226102760287},
  url          = {https://publications.ait.ac.at/en/publications/who-stores-the-private-key-an-exploratory-study-about-user-prefer}
}

@article{yu2024walletchoices,
  author       = {Yu, Yijun and Chang, Minsi and Reijers, Wessel and McAuley, Derek},
  title        = {How Cryptocurrency Users Choose and Secure Their Wallets},
  journal      = {Proceedings of the ACM on Human-Computer Interaction},
  volume       = {8},
  number       = {CSCW2},
  pages        = {1--27},
  year         = {2024},
  publisher    = {Association for Computing Machinery},
  doi          = {10.1145/3642534},
  url          = {https://dl.acm.org/doi/10.1145/3642534}
}

@article{moosavi2021lissy,
  author       = {Moosavi, Mahsa and Clark, Jeremy},
  title        = {Lissy: Experimenting with On-Chain Order Books},
  journal      = {arXiv preprint arXiv:2101.06291},
  year         = {2021},
  url          = {https://arxiv.org/abs/2101.06291}
}

@article{wallet_features_users2025,
  author       = {Kong, Bosol and Choi, Mideum and Shin, Jungwoo and Lee, Daeho},
  title        = {What Wallet Features Do Users Want for Their Cryptocurrencies? Conjoint Analysis of User Preferences in Cryptocurrency Wallets},
  journal      = {Applied Economics},
  year         = {2025},
  pages        = {1--15},
  doi          = {10.1080/00036846.2025.2536752},
  url          = {https://doi.org/10.1080/00036846.2025.2536752}
}

@article{erinle2024shared_custodial,
  author       = {Erinle, Y. and others},
  title        = {Shared-Custodial Wallet for Multi-Party Crypto-Asset Management},
  journal      = {Future Internet},
  volume       = {17},
  number       = {1},
  pages        = {7},
  year         = {2024},
  doi          = {10.3390/fi17010007},
  url          = {https://www.mdpi.com/1999-5903/17/1/7}
}
